\documentclass[envcountsame]{CSML}

\def\dOi{13(1:16)2017}
\lmcsheading%
{\dOi}
{1--25}
{}
{}
{Dec.~\phantom02, 2015}
{Mar.~28, 2017}
{}

\usepackage{ifthen}
\newcommand{\techRep}{true} %% switch here between true and false
\newcommand{\iftechrep}{\ifthenelse{\equal{\techRep}{true}}}
\usepackage{graphicx}
\usepackage{graphics}
\usepackage{epsfig}
\usepackage{latexsym}
\usepackage{amssymb}
\usepackage{amsmath}
\usepackage{stmaryrd}
\usepackage{ifthen}
\usepackage{array}
\usepackage{mathrsfs}
\usepackage{tikz}
\usepackage{hyperref}
\usetikzlibrary{arrows, automata, positioning}
\newcommand{\N}{\mathbb{N}}

\newcommand{\myparagraph}[1]{\medskip\noindent {\bf #1}}

\newcommand{\acit}{\textsc{ACIT}}
\newcommand{\NC}{\textsc{NC}}
\newcommand{\NL}{\textsc{NL}}
\newcommand{\NP}{\textsc{NP}}
\newcommand{\Ell}{\textsc{L}}
\newcommand{\Pee}{\textsc{P}}
\newcommand{\PSPACE}{\textsc{PSPACE}}
\newcommand{\forwardspace}{\mathcal{F}}
\newcommand{\backwardspace}{\mathcal{B}}
\newcommand{\ie}{\emph{i.e.}}
\newcommand{\eg}{\emph{e.g.}}

\newcommand{\A}{\mathcal{A}}%
\newcommand{\hB}{\widehat{B}}%
\newcommand{\hankelB}{\overline{B}}%
\newcommand{\nF}{\overrightarrow{n}}%
\newcommand{\hankelF}{\overline{F}}%
\newcommand{\hF}{\widehat{F}}%

\theoremstyle{plain}\newtheorem{theorem}[thm]{Theorem}
\theoremstyle{plain}\newtheorem{proposition}[thm]{Proposition}
\theoremstyle{plain}\newtheorem{lemma}[thm]{Lemma}
\theoremstyle{plain}
\theoremstyle{definition}\newtheorem{remark}[thm]{Remark}

\begin{document}

\title[Minimisation of Multiplicity Tree Automata]{Minimisation of Multiplicity Tree Automata\rsuper*}

\author[S.~Kiefer]{Stefan Kiefer}	%required
\address{University of Oxford, UK}	%required
\email{\{stefan.kiefer, ines.marusic, james.worrell\}@cs.ox.ac.uk}  %optional
%\thanks{thanks 1, optional.}	%optional

\author[I.~Maru\v{s}i\'{c}]{Ines Maru\v{s}i\'{c}}	%optional
\address{\vspace{-18 pt}}	%optional
%\email{ines.marusic@cs.ox.ac.uk}  %optional
%\thanks{thanks 2, optional.}	%optional

\author[J.~Worrell]{James Worrell}	%optional
\address{\vspace{-18 pt}}	%optional
%\email{james.worrell@cs.ox.ac.uk}  %optional
%\thanks{thanks 3, optional.}	%optional

%% etc.

%% required for running head on odd and even pages, use suitable
%% abbreviations in case of long titles and many authors:

%% mandatory lists of keywords and classifications:
\keywords{weighted automata, tree automata, minimisation, arithmetic circuit identity testing, consistent automaton}
\subjclass{Theory of computation -- Formal languages and automata theory -- Tree languages; \quad
Theory of computation -- Formal languages and automata theory -- Automata extensions -- Quantitative automata}
\titlecomment{{\lsuper*}This is a full and improved version of the FoSSaCS'15 paper with the same title. The current paper contains full proofs of all results reported there and complete definitions of all the minimisation algorithms.}
%%%%%%%%%%%%%%%%%%%%%%%%%%%%%%%%%%%%%%%%%%%%%%%%%%%%%%%%%%%%%%%%%%%%%%%%%%%

%\sloppy
%\pagestyle{plain}

\begin{abstract}
  We consider the problem of minimising the number of states in a
  multiplicity tree automaton over the field of rational numbers.  We
  give a minimisation algorithm that runs in polynomial time assuming
  unit-cost arithmetic.  We also show that a polynomial bound in the
  standard Turing model would require a breakthrough in the complexity
  of polynomial identity testing by proving that the latter problem is
  logspace equivalent to the decision version of minimisation.  The
  developed techniques also improve the state of the art in
  multiplicity word automata: we give an NC algorithm for minimising
  multiplicity word automata. Finally, we consider the minimal
  consistency problem: does there exist an automaton with a given
  number of states that is consistent with a given finite sample of
  weight-labelled words or trees?  We show that, over both words and
  trees, this decision problem is interreducible with the problem of
  deciding the truth of existential first-order sentences over the
  field of rationals---whose decidability is a longstanding open
  problem.
\end{abstract}

\maketitle

\section{Introduction} \label{sec-intro} Minimisation is a fundamental
problem in automata theory that is closely related to both learning
and equivalence testing.  In this work we analyse the complexity of
minimisation for multiplicity automata, \ie, weighted automata over a
field.  Minimisation of multiplicity and weighted automata has
numerous applications including image
compression~\cite{albert2009digital} and reducing the space complexity
of speech recognition
tasks~\cite{mohriWAtextspeech,conf/naacl/Eisner03}.

We take a comprehensive view, looking at multiplicity automata over
both words and trees and considering both function and decision
problems.  We also look at the closely-related problem of obtaining a
minimal automaton consistent with a given finite set of
observations. We characterise the complexity of these problems in
terms of arithmetic and Boolean circuit classes.  In particular, we
give relationships to longstanding open problems in arithmetic
complexity theory.

Multiplicity tree automata were first introduced by Berstel and
Reutenauer~\cite{berstel1982recognizable} under the terminology of
linear representations of a tree series. They generalise multiplicity
word automata, introduced by
Sch{\"u}tzenberger~\cite{IC::Schutzenberger1961}, which can be viewed
as multiplicity tree automata on unary trees.  The minimisation
problem for multiplicity word automata has long been known to be
solvable in polynomial time (in the Turing model)~\cite{IC::Schutzenberger1961, Tzeng92}.

In this work, we give a new procedure for computing minimal multiplicity word
automata and thereby place minimisation in \NC, improving also on a
randomised \NC\ procedure in~\cite{KieferMOWW13}.  (Recall that
$\NL \subseteq \NC \subseteq \Pee$, where {\NC} comprises those
languages having \Ell-uniform Boolean circuits of polylogarithmic
depth and polynomial size, or, equivalently, those problems solvable
in polylogarithmic time on parallel random-access machines with
polynomially many processors.)  By comparison, it is known that minimising
deterministic word automata is \NL-complete~\cite{ChoH92}, while
minimising non-deterministic word automata is
\PSPACE-complete \cite{JiangR93}. The latter result shows, in
particular, that the bounds obtained in this paper over $\mathbb{Q}$
do not apply to weighted automata over an arbitrary semi-ring, because
non-deterministic automata can be viewed as weighted automata over the
Boolean semi-ring.

Over trees, we give what is (to the best of our knowledge) the first
complexity analysis of the problem of minimising multiplicity
automata.  We present an algorithm that minimises a given multiplicity tree automaton $\A$ in time $O\left(|\A|^2 \cdot r\right)$, where $|\A|$ is the size of $\A$ and $r$ is the maximum alphabet rank, assuming unit-cost arithmetic.  This procedure can be viewed as a concrete version of the construction of a syntactic algebra of a recognisable tree series by Bozapalidis~\cite{Bozapalidis91}.  We thus place the
problem within \PSPACE\ in the conventional Turing model, since a polynomial-time decidable problem in the unit-cost model lies in \PSPACE\ (see, e.g., \cite{Allender_variablefree}).  We are
moreover able to precisely characterise the complexity of the decision
version of the minimisation problem, showing that it is logspace equivalent to
the \emph{arithmetic circuit identity testing (\acit)} problem, commonly also
called the \emph{polynomial identity testing} problem. As far as we can tell, obtaining this complexity bound requires departing from the framework of  Bozapalidis~\cite{Bozapalidis91}. The \acit\ problem is
very well studied, with a variety of randomised polynomial-time
algorithms~\cite{demillo1978probabilistic,Schwartz:1980:FPA:322217.322225,Zippel:sparsepoly}, but, as yet, no deterministic polynomial-time procedure (see~\cite{AroraBarak2}).
In previous work we have reduced equivalence testing of multiplicity
tree automata to~\acit~\cite{DBLP:conf/mfcs/MarusicW14}; the advance
here is to reduce the more general problem of minimisation also to~\acit\@. 

Lastly, we consider the problem of computing a minimal multiplicity
automaton consistent with a finite set of input-output behaviours.
This is a natural learning problem whose complexity for deterministic
finite automata was studied by Gold~\cite{Gold78}, who showed that the
problem of exactly identifying the smallest deterministic finite
automaton consistent with a set of accepted and rejected words is
\NP-hard.  For multiplicity word automata over the field~$\mathbb{Q}$,
we show that the decision version of this problem, which we call the
\emph{minimal consistency problem}, is logspace equivalent to the
problem of deciding the truth of existential first-order sentences
over the structure $(\mathbb{Q},+,\cdot,0,1)$, a longstanding open
problem (see~\cite{Pheidas:Hiblerts-tenth}). We observe that, by contrast, the minimal consistency problem for multiplicity word automata over the field~$\mathbb{R}$ is in \PSPACE, and likewise for multiplicity tree automata over $\mathbb{R}$ that have a fixed alphabet rank.

\myparagraph{Further Related Work.}  Based on a generalisation of the
Myhill-Nerode theorem to trees, one obtains a procedure for minimising
deterministic tree automata that runs in time quadratic in the size of
the input automaton~\cite{Brainerd68,CarrascoDF07}.  There have also
been several works on minimising deterministic tree automata with
weights in a semi-field (\ie, a semi-ring with multiplicative
inverses).  In particular, Maletti~\cite{Maletti09} gives a
polynomial-time algorithm in this setting, assuming unit cost for
arithmetic in the semi-field.  

In the non-deterministic case, Carme \emph{et al.}~\cite{CarmeGLTT03} define the subclass of \emph{residual finite} non-deterministic tree automata.  They show that
this class expresses the class of regular tree languages and admits a
polynomial-space minimisation procedure.

\section{Preliminaries} \label{sec-prelim}
Let $\mathbb{N}$ and~$\mathbb{N}_{0}$ denote the set of all positive and nonnegative integers, respectively. For every~$n \in \mathbb{N}$, we write~$[n]$ for the set $\{1, 2, \ldots, n\}$ and write $I_{n}$ for the identity matrix of order~$n$. For every $i \in [n]$, we write $e_i$ for the $i^\text{th}$ $n$-dimensional coordinate row vector. We write $\mathbf{0}_{n}$ for the $n$-dimensional zero row vector. 

For any matrix $A$, we write $A_{i}$ for its $i^\text{th}$ row, $A^{j}$ for its $j^\text{th}$ column, and $A_{i, j}$ for its~$(i,j)^\text{th}$ entry. Given nonempty subsets $I$ and $J$ of the rows and columns of $A$, respectively, we write~$A_{I, J}$ for the submatrix $(A_{i, j})_{i \in I, j \in J}$ of $A$.

Given a field  $\mathbb{F}$ and a set $S \subseteq \mathbb{F}^{n}$, we use $\langle S \rangle$ to denote the vector subspace of $\mathbb{F}^{n}$ that is spanned by $S$, where we often omit the braces when denoting $S$.

\subsection{Row and Column Spaces}
Let $\mathbb{F}$ be either the field of rationals $\mathbb{Q}$ or the field of reals $\mathbb{R}$. Let $A$ be an $m \times n$ matrix with entries in $\mathbb{F}$. The \emph{row space} of $A$, written as $\mathit{RS}(A)$, is the subspace of $\mathbb{F}^{n}$ spanned by the rows of $A$. The \emph{column space} of $A$, written as $\mathit{CS}(A)$, is the subspace of $\mathbb{F}^{m}$ spanned by the columns of $A$. That is, $\mathit{RS}(A) = \langle v \cdot A : v \in \mathbb{F}^{m} \rangle$ and $\mathit{CS}(A) = \langle A \cdot v^{\top} : v \in \mathbb{F}^{n} \rangle$.

The following Lemmas \ref{lemma:tree-function-min:rowsubset}-\ref{matrix:rowcolspace} contain some basic results about row and column spaces that we will use in this paper. 

\begin{lemma} \label{lemma:tree-function-min:rowsubset}
Let $A_{1}, A_{2}$ be matrices such that $\mathit{RS}(A_{1}) \subseteq \mathit{RS}(A_{2})$. For any matrix B such that $A_{1} \cdot B$ (and thus also $A_{2} \cdot B$) is defined, we have that 
\begin{align*}
\mathit{RS}(A_{1} \cdot B) \subseteq \mathit{RS}(A_{2} \cdot B).
\end{align*}
\end{lemma}
\begin{proof}
Suppose $A_{1} \in \mathbb{F}^{m_1 \times n}$ and $A_{2} \in \mathbb{F}^{m_2 \times n}$. For every vector $v_{1} \in \mathbb{F}^{m_1}$, it holds that $v_{1} \cdot A_{1} \in \mathit{RS}(A_{1}) \subseteq \mathit{RS}(A_{2})$. Hence, there exists a vector $v_{2} \in \mathbb{F}^{m_2}$ such that $v_{1} \cdot A_{1} = v_{2} \cdot A_{2}$. Thus
\begin{align*}
\mathit{RS}(A_{1} \cdot B) &= \langle v_{1} \cdot A_{1} \cdot B : v_{1} \in \mathbb{F}^{m_1} \rangle\\ &\subseteq \langle v_{2} \cdot A_{2}  \cdot B : v_{2} \in \mathbb{F}^{m_2} \rangle = \mathit{RS}(A_{2} \cdot B),
\end{align*}
which completes the proof.
\end{proof}

\begin{lemma} \label{lem-ATA-rowspace}
For any matrix $A \in \mathbb{F}^{m \times n}$, it holds that $\mathit{RS}(A^\top A) = \mathit{RS}(A)$.
\end{lemma}
\begin{proof}
For any $x \in \mathbb{F}^n$ such that $(A^\top A) x^\top = \mathbf{0}_{n}^\top$ we have
\[
 (A x^\top)^\top A x^\top \ = \ x A^\top A x^\top \ = \ x \mathbf{0}_{n}^\top \ = \ 0\;,
\]
and hence $A x^\top = \mathbf{0}_{m}^{\top}$.
Conversely, for any $x \in \mathbb{F}^n$ with $A x^\top = \mathbf{0}_{m}^{\top}$ we have $(A^\top A) x^\top = \mathbf{0}_{n}^{\top}$.
Therefore, matrices $A$ and~$A^\top A$ have the same null space and hence the same row space.
\end{proof}

\begin{lemma} \label{matrix:rowcolspace}
Let $A_1$, $A_2$, $B_1$, $B_2$ be matrices of dimension $n_1 \times m$, $n_2 \times m$, $m \times n_3$, $m \times n_4$, respectively. If $\mathit{RS}(A_{1}) = \mathit{RS}(A_{2})$ and $\mathit{CS}(B_{1}) = \mathit{CS}(B_{2})$, then
\begin{align*}
\mathit{rank}(A_{1} \cdot B_{1}) = \mathit{rank}(A_{2} \cdot B_{2}).
\end{align*}
\end{lemma}
\begin{proof}
By definition of rank as the dimension of row or column space, we have
\allowdisplaybreaks\begin{align*}
\mathit{rank}(A_{1} \cdot B_{1}) &= \mathit{dim}\; \langle x \cdot A_{1} \cdot B_{1} : x \in \mathbb{F}^{n_1} \rangle\\
&= \mathit{dim}\; \langle x \cdot A_{2} \cdot B_{1} : x \in \mathbb{F}^{n_2} \rangle && (\text{using }  \mathit{RS}(A_{1}) = \mathit{RS}(A_{2}))\\
&= \mathit{dim}\; \langle A_{2} \cdot B_{1} \cdot x^{\top} : x \in \mathbb{F}^{n_3} \rangle\\
&= \mathit{dim}\; \langle A_{2} \cdot B_{2} \cdot x^{\top} : x \in \mathbb{F}^{n_4} \rangle && (\text{using } \mathit{CS}(B_{1}) = \mathit{CS}(B_{2}))\\
&= \mathit{rank}(A_{2} \cdot B_{2}). %\tag*{$\blacksquare$}
\end{align*}
This completes the proof.
\end{proof}

\subsection{Kronecker Product} %\label{section_Kronecker}
Let $A$ be an $m_1 \times n_1$ matrix and $B$ an $m_2 \times n_2$ matrix. The \emph{Kronecker product} of $A$ by $B$, written as $A\otimes B$, is an $m_1 m_2 \times n_1 n_2$ matrix where 
\[
(A\otimes B)_{(i_1 - 1) m_2 + i_2, (j_1 - 1)  n_2 + j_2} = A_{i_1 ,j_1} \cdot B_{i_2, j_2}
\]
for every $i_1 \in [m_1]$, $i_2 \in [m_2]$, $j_1 \in [n_1]$, $j_2 \in [n_2]$.

The Kronecker product is bilinear, associative, and has the following \emph{mixed-product property}: For any matrices $A$, $B$, $C$, $D$ such that products $A \cdot C$ and $B \cdot D$ are defined, it holds that $( A \otimes B) \cdot (C \otimes D) = (A \cdot C) \otimes (B \cdot D)$.

For every $k \in \N_0$ we define the \emph{$k$-fold Kronecker power} of a matrix $A$, written as $A^{\otimes k}$, inductively by $A^{\otimes 0} = I_1$ and $A^{\otimes k} = A^{\otimes (k-1)} \otimes A$ for $k \geq 1$.

Let $k \in \mathbb{N}$, and let $n_1, \ldots, n_k \in \mathbb{N}$. Suppose $A$ is a matrix with $n_1 \cdot \ldots \cdot n_k$ rows. For every $(i_1, \ldots, i_{k}) \in  [n_1] \times\cdots \times [n_k]$, we use $A_{(i_1, \ldots, i_{k})}$ to denote the $(\sum_{l=1}^{k-1} (i_l -1 ) \cdot ( \prod_{ p= l+1}^{k} n_p ) + i_k)^\text{th}$ row of $A$. Let $A_1, \ldots, A_k$ be matrices such that for every $l \in [k]$,  $A_l$ has $n_l$ rows. It can easily be shown using induction on $k$ that for every $(i_1, \ldots, i_k) \in [n_1]  \times \cdots \times [n_k]$, 
\begin{align}
( A_{1} \otimes \cdots\otimes A_{k})_{(i_1, \ldots, i_k)} = (A_{1})_{i_1} \otimes \cdots\otimes (A_{k})_{i_k}. \label{Kronecker_indices}
\end{align}
We write $\bigotimes_{l=1}^{k} A_{l} := A_{1} \otimes \cdots\otimes A_{k}$.

For any $k \in \mathbb{N}_0$ and matrices $A_1, \ldots, A_k$  and $B_1, \ldots, B_k$ where product $A_l \cdot B_l$ is defined for every $l \in [k]$, we have
\begin{align}
(A_{1} \otimes \cdots\otimes A_{k}) \cdot (B_{1} \otimes \cdots\otimes B_{k}) = (A_{1} \cdot B_{1}) \otimes \cdots \otimes (A_{k} \cdot B_{k}). \label{Kronecker:kmixed}
\end{align}
This follows easily from the mixed-product property by induction on $k$.

\subsection{Multiplicity Word Automata}
Let $\Sigma$ be a finite alphabet and $\varepsilon$ be the empty word. The set of all words over $\Sigma$ is denoted by $\Sigma^{*}$, and the length of a word $w \in \Sigma^{*}$ is denoted by $|w|$. For any $n \in \N_{0}$ we write
$\Sigma^{n} := \{w \in \Sigma^{*} : |w| = n\}$,
$\Sigma^{\leq  n} := \bigcup_{l=0}^{n} \Sigma^{l}$, and
$\Sigma^{< n} := \Sigma^{\leq  n} \setminus \Sigma^{n}$.
Given two words $x, y \in \Sigma^{*}$, we denote by $x y$ the concatenation of $x$ and~$y$. Given two sets $X, Y \subseteq \Sigma^{*}$, we define $X Y := \{ x y : x \in X, y \in Y\}$.

Let $\mathbb{F}$ be a field. A \emph{word series} over $\Sigma$ with coefficients in $\mathbb{F}$ is a mapping $f: \Sigma^{*}  \to \mathbb{F}$. The
\emph{Hankel matrix} of $f$ is matrix $H: \Sigma^{*} \times
\Sigma^{*} \to \mathbb{F}$ such that $H_{x, y} = f(x y)$ for all
$x, y \in \Sigma^{*}$.
%The set of all word series  over $\Sigma$ with coefficients in $\mathbb{F}$  is denoted by $\mathbb{F}\langle\langle \Sigma^{*} \rangle\rangle$.

An $\mathbb{F}$-\emph{multiplicity word automaton} ($\mathbb{F}$-\emph{MWA}) is a $5$-tuple $\A = (n, \Sigma, \mu, \alpha, \gamma)$ which consists of the \emph{dimension} $n \in \N_{0}$ representing the number of states, a finite alphabet $\Sigma$, a function $\mu : \Sigma \to \mathbb{F}^{n \times n }$ assigning a \emph{transition matrix} $\mu(\sigma)$ to each $\sigma \in \Sigma$, the \emph{initial weight vector} $\alpha \in \mathbb{F}^{1 \times n}$, and the \emph{final weight vector} $\gamma \in \mathbb{F}^{n \times 1}$. We extend
the function $\mu$ from $\Sigma$ to $\Sigma^{*}$ by defining $\mu(\varepsilon) := I_{n}$, and $\mu(\sigma_{1} \cdots \sigma_{k}) := \mu(\sigma_{1}) \cdot \ldots \cdot \mu(\sigma_{k})$ for
any $\sigma_{1}, \ldots, \sigma_{k} \in \Sigma$. It is easy to see that $\mu( x y) = \mu(x) \cdot \mu(y)$ for any $x, y \in \Sigma^{*}$. Automaton $\A$ \emph{recognises} the word series $\| \A \|: \Sigma^* \to \mathbb{F}$ where ${\|\A \| (w) = \alpha \cdot \mu(w) \cdot \gamma}$ for every $w \in
\Sigma^{*}$.
%Two automata $\A_1$, $\A_2$ are said to be \emph{equivalent} if $\| \A_1 \| \equiv \| \A_2 \|$. An $\mathbb{F}$-multiplicity word automaton is said to be \emph{minimal} if no equivalent automaton has strictly smaller dimension.

\subsection{Finite Trees} A \emph{ranked alphabet} is a tuple  $(\Sigma, \mathit{rk})$ where $\Sigma$ is a nonempty finite set of symbols and $\mathit{rk}: \Sigma  \to \N_{0}$ is a function. Ranked alphabet $(\Sigma, \mathit{rk})$ is often written $\Sigma$ for short. For every $k \in \N_{0}$, we define the set of all \emph{$k$-ary} symbols $\Sigma_{k} := \mathit{rk}^{-1}(\{k\})$. We say that $\Sigma$ has \emph{rank} $r$ if $r = \max\{\mathit{rk}(\sigma) : \sigma \in \Sigma \}$.

The set of \emph{$\Sigma$-trees} (\emph{trees} for short), written as $T_{\Sigma}$, is the smallest set $T$ satisfying the following two conditions: (i) $\Sigma_{0} \subseteq T$; and  (ii) if $k \geq 1$, $\sigma \in \Sigma_{k}$, $t_{1}, \ldots, t_{k} \in T$ then $\sigma (t_{1}, \ldots, t_{k} ) \in T$.
The \emph{height} of a tree $t$, written as $\mathit{height}(t)$, is defined by $\mathit{height}(t) = 0$
if $t \in \Sigma_{0}$, and $\mathit{height}(t) = 1+ \max_{i \in [k]}
\mathit{height}(t_{i})$ if $t = \sigma (t_{1}, \ldots , t_{k} )$ for some $k
\geq 1$, $\sigma \in \Sigma_{k}$, $t_{1}, \ldots, t_{k} \in T_{\Sigma}$.
For any $n \in \N_{0}$ we write
$T_{\Sigma}^{n} := \{t \in T_{\Sigma} : \mathit{height}(t) = n\}$,
$T_{\Sigma}^{\leq  n} := \bigcup_{l=0}^{n} T_{\Sigma}^{l}$, and
$T_{\Sigma}^{< n} := T_{\Sigma}^{\leq  n} \setminus T_{\Sigma}^{n}$.

Let $\Box$ be a nullary symbol not contained in $\Sigma$. The set
$C_{\Sigma}$ of \emph{$\Sigma$-contexts} (\emph{contexts} for short)
is the set of all $(\{\Box\} \cup \Sigma)$-trees in which $\Box$ occurs
exactly once. Let $n \in \mathbb{N}_{0}$. We denote by $C_{\Sigma}^{n}$ the set of all contexts $c \in C_{\Sigma}$  where the distance between the root and the $\Box$-labelled node of $c$ is equal to $n$. Moreover, we write $C_{\Sigma}^{\leq n} := \bigcup_{l=0}^{n} C_{\Sigma}^{l}$ and $C_{\Sigma}^{< n} := C_{\Sigma}^{\leq n} \setminus C_{\Sigma}^{n}$.  A  \emph{subtree} of $c \in C_{\Sigma}$ is a $\Sigma$-tree consisting of a node in $c$ and all of its descendants.  Given a set $S \subseteq T_{\Sigma}$, we denote by $C_{\Sigma, S}^{n}$ the set of all contexts $c \in C_{\Sigma}^{n}$  where every subtree of $c$ is an element of $S$. Moreover, we write $C_{\Sigma, S}^{\leq n} := \bigcup_{l=0}^{n} C_{\Sigma, S}^{l}$ and $C_{\Sigma, S}^{< n} := C_{\Sigma, S}^{\leq n} \setminus C_{\Sigma, S}^{n}$.

Given $c \in C_{\Sigma}$ and $t \in T_{\Sigma}
\mathop{\dot{\cup}} C_{\Sigma}$, we write $c [t]$ for the tree
obtained by substituting $t$ for $\Box$ in $c$.  Let $\mathbb{F}$ be a field. A \emph{tree series} over $\Sigma$ with coefficients in $\mathbb{F}$ is a mapping $f: T_{\Sigma}  \to \mathbb{F}$. The
\emph{Hankel matrix} of $f : T_\Sigma \to \mathbb{F}$ is the matrix $H: T_{\Sigma} \times
C_{\Sigma} \to \mathbb{F}$ such that $H_{t, c} = f(c [t])$ for every
$t \in T_{\Sigma} $ and $c \in C_{\Sigma}$.

\subsection{Multiplicity Tree Automata} \label{subsec:MTA}
Let $\mathbb{F}$ be a field. An $\mathbb{F}$-\emph{multiplicity tree automaton} ($\mathbb{F}$-\emph{MTA}) is a $4$-tuple $\A = (n, \Sigma, \mu, \gamma)$ which consists of the \emph{dimension} $n \in \N_{0}$ representing the number of states, a ranked alphabet $\Sigma$, the \emph{tree representation} $\mu = \{\mu(\sigma) : \sigma \in \Sigma\}$ where for every symbol $\sigma \in \Sigma$, $\mu(\sigma)  \in \mathbb{F}^{ n^{\mathit{rk}(\sigma)} \times n }$ represents the \emph{transition matrix} associated to $\sigma$, and
the \emph{final weight vector} $\gamma \in \mathbb{F}^{n \times 1}$.
We speak of an MTA if the field~$\mathbb{F}$ is clear from the context or irrelevant.
%\footnote{We measure size assuming explicit rather than sparse representations of the transition matrices and final weight vector because minimal automata are only unique up to change of basis; cf. Theorem \ref{thm:minimal_unique}.}
The \emph{size} of $\A$, written as $|\A|$, is the total number of entries in all transition matrices and the final weight
vector of $\A$, \ie, $|\A| := \sum_{\sigma \in \Sigma}
n^{\mathit{rk}(\sigma) + 1} + n$.

We extend
the tree representation $\mu$ from $\Sigma$ to $T_{\Sigma}$ by defining 
\[
\mu(\sigma(t_{1}, \ldots, t_{k})) :=
(\mu(t_{1}) \otimes \cdots \otimes \mu(t_{k}) ) \cdot \mu(\sigma)
\]
for every $\sigma \in \Sigma_{k}$ and $t_{1}, \ldots,t_{k} \in
T_{\Sigma}$. Automaton $\A$ \emph{recognises} the tree series $\| \A \| : T_{\Sigma} \to \mathbb{F}$ where ${\|\A \| (t) = \mu(t) \cdot \gamma}$ for every $t \in
T_{\Sigma}$.

We further extend $\mu$ from $T_{\Sigma}$ to $C_{\Sigma}$ by treating $\Box$ as a unary symbol and defining $\mu (\Box) := I_{n}$. This allows to define $\mu(c) \in \mathbb{F}^{n \times n}$ for every $c = \sigma (t_{1}, \ldots, t_{k}) \in C_{\Sigma}$ inductively as $\mu(c)  :=  \left(\mu(t_{1}) \otimes \cdots \otimes \mu(t_{k}) \right) \cdot \mu(\sigma)$. It is easy to see that
for every $t \in T_{\Sigma}
\mathop{\dot{\cup}} C_{\Sigma}$ and $c \in C_{\Sigma}$, $\mu(c [t]) = \mu(t) \cdot \mu(c)$.

MWAs can be seen as a special case of MTAs:
An MWA $(n, \Sigma, \mu, \alpha, \gamma)$ ``is''
 the MTA $(n, \Sigma \mathop{\dot{\cup}} \{\sigma_0\}, \mu, \gamma)$
 where the symbols in~$\Sigma$ are unary, symbol $\sigma_0$ is nullary,
 and $\mu(\sigma_0) = \alpha$.
That is, we view $\left(\Sigma \mathop{\dot{\cup}} \{\sigma_0\}\right)$-trees as words over~$\Sigma$ by
 omitting the leaf symbol~$\sigma_0$.
Hence if a result holds for MTAs, it also holds for MWAs.
%Note that the definitions of a Hankel matrix for words and trees are compatible.
Some concepts, such as contexts, would formally need adaptation,
 however we omit such adaptations as they are straightforward.
Therefore, we freely view MWAs as MTAs whenever convenient.

Two MTAs $\A_1$, $\A_2$ are said to be \emph{equivalent} if $\| \A_1 \| = \| \A_2 \|$.
An MTA is said to be \emph{minimal} if no equivalent automaton has strictly smaller dimension.
The following result was first shown by Habrard and Oncina~\cite{habrardlearning}, although a closely-related result was given by Bozapalidis and Louscou-Bozapalidou~\cite{bozapalidis1983rank}.

\begin{theorem}[\cite{bozapalidis1983rank,habrardlearning}]   \label{thm:Hankel} 
Let $\Sigma$ be a ranked alphabet, $\mathbb{F}$ be a field, and $f : T_{\Sigma} \to \mathbb{F}$. Let $H$ be the Hankel matrix of $f$.
Then, $f$ is recognised by some MTA if and only if $H$ has finite rank over $\mathbb{F}$. In case $H$ has finite rank over $\mathbb{F}$, the dimension of a minimal MTA recognising $f$ is $\mathit{rank}(H)$ over $\mathbb{F}$. \qed
\end{theorem}
It follows from Theorem \ref{thm:Hankel} that an $\mathbb{F}$-MTA $\A$ of dimension $n$ is minimal if and only if the Hankel matrix of $\| \A \|$ has rank $n$ over $\mathbb{F}$.

\begin{remark}\label{rem:Hankelword}
Theorem~\ref{thm:Hankel} specialised to word automata was proved by Carlyle and Paz~\cite{carlyle1971realizations} and Fliess~\cite{fliess}.
Their proofs show that if
$X,Y\subseteq\Sigma^*$ are such that
$\mathit{rank}(H_{X,Y})=\mathit{rank}(H)$, then $f$ is uniquely
determined by $H_{X,Y}$ and $H_{X\Sigma,Y}$.
\end{remark}

In the remainder of this section, we prove some closure properties for MTAs.  First, we give two definitions: the product and the difference of two $\mathbb{F}$-MTAs.
Let $\A_{1} = (n_{1},\Sigma, \mu_{1}, \gamma_{1})$ and $\A_{2} =
(n_{2},\Sigma, \mu_{2}, \gamma_{2})$ be two $\mathbb{F}$-multiplicity tree automata. The \emph{difference} of $\A_{1}$ and $\A_{2}$, written as
$\A_{1} - \A_{2}$, is the $\mathbb{F}$-multiplicity tree automaton $(n, \Sigma, \mu, \gamma)$ where:
\begin{itemize}
\item[$\bullet$] $n = n_1 + n_2$;
\item[$\bullet$] For every $\sigma \in \Sigma$ and any $i \in [(n_{1} + n_{2})^{\mathit{rk}(\sigma)}]$, $j \in [n_{1} + n_{2}]$, \begin{align*}
\mu(\sigma)_{i,j}=
\begin{cases}
\mu_{1}(\sigma)_{i,j}  & \text{if } i \leq n_{1}^{\mathit{rk}(\sigma)} \text{and } j \leq n_1 \\
\mu_{2}(\sigma)_{i,j}  & \text{if } i > (n_{1} + n_{2})^{\mathit{rk}(\sigma)} - n_{2}^{\mathit{rk}(\sigma)} \text{and } j > n_1 \\
0 & \text{otherwise;}
\end{cases}
\end{align*}
\item[$\bullet$]$\gamma = \begin{bmatrix}
 \gamma_1\\  -\gamma_2
 \end{bmatrix}$.
\end{itemize}
The \emph{product} of $\A_{1}$ by $\A_{2}$, written as
$\A_{1} \times \A_{2}$, is the $\mathbb{F}$-multiplicity tree automaton $(n, \Sigma, \mu, \gamma)$ where:
\begin{itemize}
\item[$\bullet$] $n = n_{1} \cdot n_{2}$;
\item[$\bullet$] For every $\sigma \in \Sigma_{k}$, $\mu (\sigma) = P_{k} \cdot (\mu_{1}(\sigma)\otimes \mu_{2}(\sigma))$ where $P_{k}$ is a permutation matrix of order $(n_{1} \cdot n_{2})^{k}$ uniquely defined (see Remark \ref{rmk:productMTA-defn} below) by
\begin{align} \label{MTACartesian_transition}
(u_{1} \otimes \cdots \otimes u_{k}) \otimes (v_{1} \otimes \cdots \otimes v_{k}) = ((u_{1} \otimes v_{1}) \otimes \cdots \otimes (u_{k} \otimes v_{k})) \cdot P_{k}
\end{align}
for all $u_{1}, \ldots, u_{k} \in \mathbb{F}^{1 \times n_{1}}$ and $v_{1}, \ldots, v_{k} \in \mathbb{F}^{1 \times n_{2}}$;
\item[$\bullet$] $\gamma = \gamma_{1} \otimes \gamma_{2}$.
\end{itemize}

\begin{remark} \label{rmk:productMTA-defn}
In the following we argue that for every $k$, matrix $P_{k}$ is well-defined by Equation (\ref{MTACartesian_transition}). To do this, it suffices to show that $P_{k}$ is well-defined on a set of basis vectors of $\mathbb{F}^{1 \times n_{1}}$ and $\mathbb{F}^{1 \times n_{2}}$ and then extend linearly. To that end, let $(e_{i}^{1})_{i \in [n_{1}]}$ and $(e_{j}^{2})_{j \in [n_{2}]}$ be bases of $\mathbb{F}^{1 \times n_{1}}$ and $\mathbb{F}^{1 \times n_{2}}$, respectively. Then
\begin{align*}
E_{1} := \{(e_{i_1}^{1} \otimes \cdots \otimes e_{i_k}^{1}) \otimes (e_{j_1}^{2} \otimes \cdots \otimes e_{j_k}^{2}) : i_1, \ldots, i_k \in [n_{1}], j_1, \ldots, j_k \in [n_{2}]\}
\end{align*}
and
\begin{align*}
E_{2} := \{(e_{i_1}^{1} \otimes e_{j_1}^{2}) \otimes \cdots \otimes (e_{i_k}^{1} \otimes e_{j_k}^{2}) : i_1, \ldots, i_k \in [n_{1}], j_1, \ldots, j_k \in [n_{2}]\}
\end{align*}
are two bases of the vector space $\mathbb{F}^{1 \times n_{1} n_{2}}$. Therefore, $P_{k}$ is well-defined as an invertible matrix mapping basis $E_{1}$ to basis $E_{2}$.
\end{remark}

We now turn to the closure properties for MTAs:

\begin{proposition}\label{prop:diff_prod_MTA}
Let $\A_1 = (n_1, \Sigma, \mu_{1}, \gamma_{1})$ and $\A_2 = (n_2, \Sigma, \mu_{2}, \gamma_{2})$ be two $\mathbb{F}$-MTAs.
For their difference $\A_{1} - \A_{2}$, it holds that $\| \A_{1} - \A_{2} \| = \| \A_{1} \| - \| \A_{2} \|$.  For their product $\A_{1} \times \A_{2} = (n, \Sigma, \mu, \gamma)$, the following properties hold:
\begin{enumerate}[label=(\roman*)]
\item for every $t \in T_{\Sigma}$, $\mu(t) = \mu_{1}(t) \otimes \mu_{2}(t)$;
\item for every $c \in C_{\Sigma}$, $\mu(c) = \mu_{1}(c) \otimes \mu_{2}(c)$;
\item $\| \A_{1} \times \A_{2}\| =  \| \A_{1} \| \cdot \|\A_{2}\|$.
\end{enumerate}
When $\mathbb{F} = \mathbb{Q}$, both automata $\A_{1} - \A_{2}$ and $\A_{1} \times \A_{2}$ can be computed from $\A_{1}$ and $\A_{2}$ in logarithmic space.
\end{proposition}
\begin{proof} The result for the difference automaton is shown in \cite[Proposition 3.1]{berstel1982recognizable}. Results (i) and (iii) for the product automaton are shown in \cite[Proposition 5.1]{berstel1982recognizable}; see also \cite{DBLP:journals/actaC/Borchardt04}. In the following we prove the remainder of the proposition.

We prove result (ii) using induction on the distance between the root and the $\Box$-labelled node of $c$. The base case is $c = \Box$. Here by definition we have that
\begin{align*}
\mu(c) = \mu(\Box) = I_{n_{1} \cdot n_{2}} = I_{n_{1}} \otimes I_{n_{2}} = \mu_{1}(\Box) \otimes \mu_{2}(\Box) = \mu_{1}(c) \otimes \mu_{2}(c).
\end{align*}
For the induction step, let $h \in \N_{0}$ and assume that (ii) holds for every context $c \in C_{\Sigma}^{h}$. Take any $c \in C_{\Sigma}^{h+1}$. Without loss of generality we can assume that $c = \sigma(c_{1}, t_{2}, \ldots, t_{k})$ for some $k \geq 1$, $\sigma \in \Sigma_{k}$, $c_{1} \in C_{\Sigma}^{h}$, and $t_2, \ldots, t_k \in T_{\Sigma}$. By the induction hypothesis, result (i), Equation (\ref{MTACartesian_transition}), and the mixed-product property of Kronecker product, we now have
\allowdisplaybreaks \begin{align*}
\mu(c) & = \left( \mu(c_1) \otimes \bigotimes_{j=2}^{k} \mu(t_j)\right) \cdot \mu(\sigma)\\
& = \left( (\mu_{1}(c_{1}) \otimes \mu_{2}(c_{1})) \otimes \bigotimes_{j=2}^{k} (\mu_{1}(t_{j}) \otimes \mu_{2}(t_{j}))\right) \cdot P_{k} \cdot (\mu_{1}(\sigma)\otimes \mu_{2}(\sigma))\\
& = \left(\left(\mu_{1}(c_{1}) \otimes \bigotimes_{j=2}^{k} \mu_{1}(t_j)\right) \otimes \left(\mu_{2}(c_{1}) \otimes \bigotimes_{j=2}^{k} \mu_{2}(t_j)\right)\right) \cdot (\mu_{1}(\sigma)\otimes \mu_{2}(\sigma))\\
& = \left(\left(\mu_{1}(c_{1}) \otimes \bigotimes_{j=2}^{k} \mu_{1}(t_j)\right) \cdot \mu_{1}(\sigma)\right) \otimes \left(\left(\mu_{2}(c_{1}) \otimes \bigotimes_{j=2}^{k} \mu_{2}(t_j)\right) \cdot \mu_{2}(\sigma)\right)\\
& = \mu_{1}(c) \otimes \mu_{2}(c).
\end{align*}
This completes the proof of result (ii) by induction.

Now let $\mathbb{F} = \mathbb{Q}$. The $\mathbb{Q}$-MTA $\A_{1} \times \A_{2}$ can be computed using a deterministic Turing machine which scans the transition matrices and the final weight vectors of $\A_{1}$ and $\A_{2}$, and then writes down the entries of the transition matrices and the final weight vector of their product $\A_{1} \times \A_{2}$ onto the output tape. This computation requires maintaining only a constant number of pointers, which takes logarithmic space in the representation of automata $\A_{1}$ and $\A_{2}$. Hence, the Turing machine computing the automaton $\A_{1} \times \A_{2}$ uses logarithmic space in the work tape. Analogously, the $\mathbb{Q}$-MTA $\A_{1} - \A_{2}$ can be computed from $\A_{1}$ and $\A_{2}$ in logarithmic space.
\end{proof}

\section{Fundamentals of Minimisation} \label{sec-foundations}

In this section, we prepare the ground for minimisation algorithms.
Let us fix a field~$\mathbb{F}$ for the rest of this section and assume that all automata are over~$\mathbb{F}$.
We also fix an MTA $\A = (n, \Sigma, \mu, \gamma)$ for the rest of the section.
We will construct from $\A$ another MTA~$\tilde{\A}$ which we show to be equivalent to~$\A$ and minimal.
A crucial ingredient for this construction are special vector spaces induced by~$\A$,
 called the forward space and the backward space.

\subsection{Forward and Backward Space} \label{sub-forward-backward-definition}

The \emph{forward space~$\forwardspace$} of~$\A$ is the (row) vector space
 $\forwardspace := \langle\mu(t) : t \in T_{\Sigma}\rangle$ over~$\mathbb{F}$.
The \emph{backward space~$\backwardspace$} of~$\A$ is the (column) vector space
 $\backwardspace := \langle\mu(c) \cdot \gamma : c \in C_{\Sigma}\rangle$ over~$\mathbb{F}$.
The following Propositions \ref{prop:forwardspace_charact}~and~\ref{prop:backwardspace_charact} provide fundamental characterisations of $\forwardspace$~and~$\backwardspace$, respectively.
\newcommand{\stmtpropforwardspacecharact}{
The forward space $\forwardspace$ has the following properties:
\begin{enumerate}
\renewcommand{\theenumi}{(\alph{enumi})}
\renewcommand{\labelenumi}{\textup{\theenumi}}
\item The forward space $\forwardspace$ is the smallest vector space $V$ over $\mathbb{F}$ such that for all $k \in \N_0$, $v_{1}, \ldots, v_{k} \in V$, and  $\sigma \in \Sigma_{k}$ it holds that
 $(v_{1} \otimes \cdots \otimes v_{k}) \cdot \mu (\sigma) \in V$.
\item The set of row vectors $\{\mu (t) : t \in T_{\Sigma}^{< n}\}$ spans $\forwardspace$.
\end{enumerate}
}
\begin{proposition} \label{prop:forwardspace_charact}
\stmtpropforwardspacecharact
\end{proposition}
\begin{proof}
We start by proving result~(a). Here we first show that $\forwardspace$ has the closure property stated in~(a). To this end, let us take any $k \in \N_0$, $v_{1}, \ldots, v_{k} \in \forwardspace$, and $\sigma \in \Sigma_{k}$.
By definition of the forward space $\forwardspace$, for every $i \in [k]$ we can express vector $v_{i} \in \forwardspace$ as
\begin{align*}
v_{i} = \sum\limits_{j_{i} = 1}^{m_{i}} \alpha_{j_{i}}^{i} \mu (t_{j_{i}}^{i})
\end{align*}
for some  integer $m_{i} \in \N$, scalars $\alpha_{1}^{i}, \ldots, \alpha_{m_{i}}^{i} \in \mathbb{F}$, and trees $t_{1}^{i}, \ldots, t_{m_{i}}^{i} \in T_{\Sigma}$.  From here, using bilinearity of Kronecker product we get that
\allowdisplaybreaks \begin{align*}
(v_{1} \otimes \cdots \otimes v_{k}) \cdot \mu (\sigma) &=
\left(\left(\sum\limits_{j_{1} = 1}^{m_{1}} \alpha_{j_{1}}^{1} \mu (t_{j_{1}}^{1})\right) \otimes \cdots \otimes \left(\sum\limits_{j_{k} = 1}^{m_{k}} \alpha_{j_{k}}^{k} \mu (t_{j_{k}}^{k})\right)\right) \cdot \mu (\sigma)\\
&= \sum\limits_{j_{1} = 1}^{m_{1}} \cdots \sum\limits_{j_{k} = 1}^{m_{k}} \alpha_{j_{1}}^{1} \cdots \alpha_{j_{k}}^{k} \left(\mu (t_{j_{1}}^{1}) \otimes \cdots \otimes \mu (t_{j_{k}}^{k})\right) \cdot \mu (\sigma)\\
&= \sum\limits_{j_{1} = 1}^{m_{1}} \cdots \sum\limits_{j_{k} = 1}^{m_{k}} \alpha_{j_{1}}^{1} \cdots \alpha_{j_{k}}^{k} \cdot
\mu (\sigma (t_{j_{1}}^{1}, \ldots, t_{j_{k}}^{k})).
\end{align*}
Since $\forwardspace$ is a vector space, the above equation implies that $(v_{1} \otimes \cdots \otimes v_{k}) \cdot \mu (\sigma) \in \forwardspace$.

Let $V$ be any vector space over~$\mathbb{F}$ such that for all $k \in \N_0$, $v_{1}, \ldots, v_{k} \in V$, and  $\sigma \in \Sigma_{k}$ it holds that
 $(v_{1} \otimes \cdots \otimes v_{k}) \cdot \mu (\sigma) \in V$.
We claim that $\forwardspace \subseteq V$.
To prove this, it suffices to show that $\mu (t) \in V$ for every $t \in T_{\Sigma}$.
Here we give a proof by induction on $\mathit{height}(t)$.
The base case $t \in \Sigma_{0}$ is trivial.
For the induction step, let $h \in \N_{0}$ and assume that $\mu (t) \in V$  for all $t \in T_{\Sigma}^{\leq h}$. Take any $t \in T_{\Sigma}^{h+1}$. Then, $t = \sigma (t_{1}, \ldots, t_{k})$ for some $k \geq 1$, $\sigma \in \Sigma_{k}$, and $t_{1}, \ldots, t_{k} \in T_{\Sigma}^{\leq h}$. The induction hypothesis now implies that $\mu (t_{1}), \ldots, \mu (t_{k}) \in V$. By the choice of~$V$, we therefore have that \mbox{$\mu (t) = (\mu (t_{1}) \otimes \cdots \otimes \mu (t_{k})) \cdot \mu (\sigma) \in V$}. This completes the proof by induction.

The proof of result (b) follows from~\cite[Main Lemma 4.1]{seidlfull}.
\end{proof}

\begin{proposition} \label{prop:backwardspace_charact}
Let $S \subseteq T_{\Sigma}$ be a set of trees such that $\{\mu (t) : t \in S\}$ spans $\forwardspace$. Then, the following properties hold:
%The backward space $\backwardspace$ has the following properties:
\begin{enumerate}
\renewcommand{\theenumi}{(\alph{enumi})}
\renewcommand{\labelenumi}{\textup{\theenumi}}
\item The backward space $\backwardspace$ is the smallest vector space $V$ over $\mathbb{F}$ such that:
\begin{enumerate}
\renewcommand{\theenumi}{(\alph{enumi})}
\renewcommand{\labelenumi}{\textup{\theenumi}}
\renewcommand{\theenumii}{(\arabic{enumii})}
\renewcommand{\labelenumii}{\textup{\theenumii}}
\item $\gamma \in V$.
\item For every $v \in V$ and $c \in C_{\Sigma, S}^{1}$ it holds that $\mu(c) \cdot v \in V$.
\end{enumerate}
\item
The set of column vectors $\{\mu (c) \cdot \gamma : c \in C_{\Sigma, S}^{< n}\}$ spans $\backwardspace$.
\end{enumerate}
\end{proposition}
\begin{proof}
First, we prove result~(a). We have that $\gamma = \mu (\Box) \cdot \gamma \in \backwardspace$, hence $\backwardspace$ satisfies property~1.
To see that~$\backwardspace$ satisfies property 2, let us take any $v \in \backwardspace$ and $c \in C_{\Sigma, S}^{1}$. By definition of~$\backwardspace$, the vector $v$ can be expressed as
\begin{align*}
v = \sum\limits_{i=1}^{m} \alpha_{i} \cdot \mu(c_{i}) \cdot \gamma
\end{align*}
for some integer $m \in \N$, scalars $\alpha_{1}, \ldots, \alpha_{m} \in \mathbb{F}$, and contexts $c_{1}, \ldots, c_{m} \in C_{\Sigma}$. Thus by bilinearity of matrix multiplication we have
\allowdisplaybreaks \begin{align*}
\mu(c) \cdot v &= \mu(c) \cdot \left(\sum\limits_{i=1}^{m} \alpha_{i} \cdot \mu(c_{i}) \cdot\gamma \right)\\ &= \sum\limits_{i=1}^{m} \alpha_{i} \cdot (\mu(c) \cdot \mu(c_{i}) \cdot \gamma)  = \sum\limits_{i=1}^{m} \alpha_{i} \cdot  \mu(c_{i}[c])  \cdot \gamma,
\end{align*}
which implies that $\mu(c) \cdot v  \in \backwardspace$ since $\backwardspace$ is a vector space. Therefore, $\backwardspace$ satisfies properties 1 and 2.

Let now $V$ be any vector space over $\mathbb{F}$ satisfying properties 1 and 2. In order to show that $\backwardspace \subseteq V$, it suffices to show that $\mu (c) \cdot \gamma \in V$ for every $c \in C_{\Sigma}$.
We prove the latter result using induction on the distance between the root and the $\Box$-labelled node of $c$.
For the induction basis, let the distance be $0$, \ie, $c = \Box$. Then we have $\mu(c) \cdot \gamma = \gamma \in V$ by property~1. For the induction step, let $h \in \N_{0}$ and assume that $\mu (c) \cdot \gamma \in V$ for all $c \in C_{\Sigma}^{\le h}$. Take any $c \in C_{\Sigma}^{h+1}$. Let $c^{\prime} \in C_{\Sigma}^{1}$ and $c^{\prime\prime} \in C_{\Sigma}^{h}$ be such that $c = c^{\prime\prime} [c^{\prime}]$. Without loss of generality we can assume that $c^{\prime} = \sigma ( \Box, \tau_{2}, \ldots, \tau_{k})$ where $k \ge 1$, $\sigma \in \Sigma_{k}$, and $\tau_{2}, \ldots, \tau_{k} \in T_{\Sigma}$. Since $\forwardspace = \langle \mu (t) : t \in S \rangle$, for every $i \in \{2, \ldots, k\}$ there is an integer~$m_{i} \in \N$, scalars $\alpha_{1}^{i}, \ldots, \alpha_{m_{i}}^{i} \in \mathbb{F}$, and trees $t_{1}^{i}, \ldots, t_{m_{i}}^{i} \in S$ such that
\begin{align*}
\mu (\tau_{i}) = \sum\limits_{j_{i} = 1}^{m_{i}} \alpha_{j_{i}}^{i} \mu (t_{j_{i}}^{i}).
\end{align*}
From here, using bilinearity of Kronecker product, it follows that
\allowdisplaybreaks \begin{align*}
\mu (c) \cdot \gamma &= \mu (c^{\prime}) \cdot \mu (c^{\prime\prime}) \cdot \gamma\\ &= (I_{n} \otimes \mu (\tau_{2}) \otimes \cdots \otimes \mu(\tau_{k}))  \mu (\sigma) \cdot \mu (c^{\prime\prime}) \cdot \gamma\\
&= \left(I_{n} \otimes \left(\sum\limits_{j_{2} = 1}^{m_{2}} \alpha_{j_{2}}^{2} \mu (t_{j_{2}}^{2})\right)  \otimes \cdots \otimes \left(\sum\limits_{j_{k} = 1}^{m_{k}} \alpha_{j_{k}}^{k} \mu (t_{j_{k}}^{k})\right)\right) \mu (\sigma) \cdot \mu (c^{\prime\prime}) \cdot \gamma\\
&= \sum\limits_{j_{2} = 1}^{m_{2}} \cdots \sum\limits_{j_{k} = 1}^{m_{k}} \alpha_{j_{2}}^{2} \cdots \alpha_{j_{k}}^{k} \cdot \left(I_{n} \otimes \mu (t_{j_{2}}^{2})  \otimes \cdots \otimes  \mu (t_{j_{k}}^{k})\right) \mu (\sigma) \cdot \mu (c^{\prime\prime}) \cdot \gamma\\
&= \sum\limits_{j_{2} = 1}^{m_{2}} \cdots \sum\limits_{j_{k} = 1}^{m_{k}} \alpha_{j_{2}}^{2} \cdots \alpha_{j_{k}}^{k}\cdot \mu (\sigma (\Box, t_{j_{2}}^{2}, \ldots, t_{j_{k}}^{k})) \cdot \mu (c^{\prime\prime}) \cdot \gamma\,,
\end{align*}
where we note that $\sigma (\Box, t_{j_{2}}^{2}, \ldots, t_{j_{k}}^{k}) \in C_{\Sigma, S}^{1}$ for every $j_{2} \in [m_{2}], \ldots, j_{k} \in [m_{k}]$.
Moreover, we have $\mu (c^{\prime\prime}) \cdot \gamma \in V$ by the induction hypothesis.
Thus by property 2 we have $\mu (\sigma (\Box, t_{j_{2}}^{2}, \ldots, t_{j_{k}}^{k})) \cdot \mu (c^{\prime\prime}) \cdot \gamma \in V$ for every $j_{2} \in [m_{2}], \ldots, j_{k} \in [m_{k}]$. Since $V$ is a vector space, we conclude that $\mu (c) \cdot \gamma \in V$.
This completes the proof of result (a) by induction.

We denote by $C_{\Sigma, S}$ the set of all $c \in C_{\Sigma}$  where every subtree of $c$ is an element of $S$. It follows easily from part (a) that $\langle \mu (c) \cdot \gamma : c \in C_{\Sigma, S}\rangle = \backwardspace$ since $\langle \mu (c) \cdot \gamma : c \in C_{\Sigma, S}\rangle$ satisfies properties 1 and 2.
Thus in order to prove result (b), it suffices to show that
 the set $\{\mu (c) \cdot \gamma : c \in C_{\Sigma, S}^{<n}\}$ spans $\langle \mu (c) \cdot \gamma : c \in C_{\Sigma, S} \rangle$.
We show this using an argument that was similarly given, \eg, in~\cite{Paz71}.
If $\gamma$ is the zero vector $\mathbf{0}_{n}^\top$, the statement is trivial.
Let us now assume that $\gamma \neq \mathbf{0}_{n}^\top$.
For every $i \in \N$, we define the vector space $\backwardspace^{i} := \langle \mu (c) \cdot \gamma : c \in C_{\Sigma, S}^{<i} \rangle$ over~$\mathbb{F}$.
Since $\backwardspace^{i}$ is a subspace of~$\backwardspace^{i+1}$ for every~$i \in \N$, we have
\begin{align}
  1 \ \le \ \mathit{dim}( \backwardspace^1) \ \le \ \mathit{dim}( \backwardspace^2) \ \le \ \cdots \ \le \ \mathit{dim}( \backwardspace^{n+1}) \ \le \ n\;, \label{backwardspace:dim_inequalitychain}
\end{align}
where the first inequality holds because $\gamma \neq \mathbf{0}_{n}^\top$,
and the last inequality holds because $\backwardspace^i \subseteq \mathbb{F}^n$ for all~$i \in \N$. Not all inequalities in the inequality chain (\ref{backwardspace:dim_inequalitychain}) can be strict, so we must have $\backwardspace^{i_0} = \backwardspace^{i_0+1}$ for some $i_0 \in [n]$.
We claim that $\backwardspace^{i} = \backwardspace^{i+1}$ for all $i \ge i_0$. We give a proof by induction on~$i$. The base case $i=i_0$ holds by definition of~$i_0$.
For the induction step, let $i \ge i_0$ and assume that $\backwardspace^{i} = \backwardspace^{i+1}$. Note that, by definition, for all $j \in \N$ we have
$\backwardspace^{j+1} = \langle \gamma, \mu (c) \cdot \backwardspace^{j}  : c \in C_{\Sigma, S}^{1}\rangle$. Using this result for $j \in \{i, i+1\}$, we obtain:
\begin{align*}
\backwardspace^{i+1}
\ = \ \langle \gamma, \mu (c) \cdot \backwardspace^{i}  : c \in C_{\Sigma, S}^{1}\rangle
\ = \ \langle \gamma, \mu (c) \cdot \backwardspace^{i+1}  : c \in C_{\Sigma, S}^{1}\rangle
\ = \ \backwardspace^{i+2}
\end{align*}
where the middle equation holds by the induction hypothesis.
This completes the proof by induction, and we thus conclude that $\backwardspace^{i} = \backwardspace^{i+1}$ for all $i \ge i_0$.
Since $n \ge i_{0}$, it follows that $\backwardspace^{n} = \bigcup_{i \ge n}\backwardspace^{i}$. Since $(\backwardspace^{i})_{i \in \N}$ is an increasing sequence of vector spaces, we have $\backwardspace^n = \bigcup_{i \in \N}\backwardspace^{i} = \langle \mu (c) \cdot \gamma : c \in C_{\Sigma, S}\rangle$ as required.
\end{proof}

\subsection{A Minimal Automaton} \label{sub-a-minimal-automaton}
Let $F$~and~$B$ be matrices whose rows and columns, respectively, span $\forwardspace$~and~$\backwardspace$. That is, $\mathit{RS}(F) = \forwardspace$ and $\mathit{CS}(B) = \backwardspace$.
We discuss later (Section \ref{function_problem:trees}) how to efficiently compute $F$~and~$B$.
The following lemma states that $\mathit{rank}(F \cdot B)$ is the dimension of a minimal automaton equivalent to $\A$.
\begin{lemma} \label{lem-rankFB}
 A minimal automaton equivalent to $\A$ has $m := \mathit{rank}(F \cdot B)$ states.
\end{lemma}
\begin{proof}
Let $H$ be the Hankel matrix of $\| \A \|$.
Define the matrix $\hankelF  \in \mathbb{F}^{T_{\Sigma} \times [n]}$
 where $\hankelF_{t} = \mu (t)$ for every $t \in T_{\Sigma}$.
Define the matrix $\hankelB \in \mathbb{F}^{[n] \times C_{\Sigma}}$
 where $\hankelB^{c} = \mu (c) \cdot \gamma$ for every $c \in C_{\Sigma}$.
For every $t \in T_{\Sigma}$ and $c \in C_{\Sigma}$ we have by the definitions that
\begin{align*}
H_{t, c} = \| \A \| (c [t]) = \mu (c [t]) \cdot \gamma =  \mu(t) \cdot \mu (c) \cdot \gamma
= \hankelF_{t} \cdot \hankelB^{c}\;,
\end{align*}
hence $H = \hankelF \cdot \hankelB$. Note that
\begin{equation}
 \mathit{RS}(\hankelF) = \forwardspace = \mathit{RS}(F) \qquad \text{ and }  \qquad
 \mathit{CS}(\hankelB) = \backwardspace = \mathit{CS}(B)\,. \label{eq-prop-decision-upper}
\end{equation}
We now have
$%\begin{equation*}
m = \mathit{rank}(H) = \mathit{rank}(\hankelF \cdot \hankelB)
 = \mathit{rank}(F \cdot B)%\label{eq-prop-decision-upper-2}
$, %\end{equation*}
where the first equality holds by Theorem~\ref{thm:Hankel} and the last equality holds by~\eqref{eq-prop-decision-upper} and Lemma~\ref{matrix:rowcolspace}.
\end{proof}

%Let $\tilde{F} \in \mathbb{F}^{m \times n}$ be a matrix
% whose rows are a (minimal) subset of the rows of~$F$ such that
% $\mathit{RS}(\tilde{F} \cdot B) = \mathit{RS}(F \cdot B)$.
Since $m = \mathit{rank}(F \cdot B)$, there exist $m$ rows of $F \cdot B$ that span $\mathit{RS}(F \cdot B)$.
The corresponding $m$ rows of~$F$ form a matrix $\tilde{F} \in \mathbb{F}^{m \times n}$ with $\mathit{RS}(\tilde{F} \cdot B) = \mathit{RS}(F \cdot B)$.
Define a multiplicity tree automaton $\tilde{\A} = (m, \Sigma, \tilde{\mu}, \tilde{\gamma})$
 with $\tilde{\gamma} = \tilde{F} \cdot \gamma$ and
\begin{align}
\tilde{\mu} (\sigma) \cdot \tilde{F} \cdot B = \tilde{F}^{\otimes k} \cdot \mu (\sigma) \cdot B
 \qquad \text{for every $\sigma \in \Sigma_{k}$.}
\label{minimimisingMTA:defn_transition}
\end{align}
We show that $\tilde{\A}$ minimises~$\A$:
\newcommand{\stmtpropminimalMTAabstract}{
The MTA~$\tilde{\A}$ is well-defined and is a minimal automaton equivalent to~$\A$.
}
\begin{proposition} \label{prop-minimal-MTA-abstract}
\stmtpropminimalMTAabstract
\end{proposition}
Before giving a full proof of Proposition~\ref{prop-minimal-MTA-abstract} later in this subsection, we now prove this result for multiplicity \emph{word} automata, stated as Proposition~\ref{prop-minimal-MWA-abstract} below, which will be used in Section~\ref{sub-word-min}.
 The main arguments are similar for the tree case, but slightly more involved.

Let $\A = (n, \Sigma, \mu, \alpha, \gamma)$ be an MWA.
The forward and backward space can then be written as $\forwardspace = \langle\alpha \cdot \mu(w) : w \in \Sigma^{*}\rangle$ and $\backwardspace = \langle\mu(w) \cdot \gamma : w \in \Sigma^{*}\rangle$, respectively.
The MWA~$\tilde{\A}$ can be written as
$\tilde{\A} = (m, \Sigma, \tilde{\mu}, \tilde{\alpha}, \tilde{\gamma})$
with $\tilde{\gamma} = \tilde{F} \cdot \gamma$,
\begin{align}
\tilde{\alpha} \cdot \tilde{F} \cdot B &= \alpha \cdot B, && \text{and} \label{minimimisingMWA:defn_alpha} \\
\tilde{\mu} (\sigma) \cdot \tilde{F} \cdot B & = \tilde{F} \cdot \mu (\sigma) \cdot B
 && \text{for every $\sigma \in \Sigma$.}
\label{minimimisingMWA:defn_transition}
\end{align}
\begin{proposition} \label{prop-minimal-MWA-abstract}
The MWA~$\tilde{\A}$ is well-defined and is a minimal automaton equivalent to~$\A$.
\end{proposition}

First, we show that $\tilde{\A}$ is a well-defined multiplicity word automaton:
\begin{lemma} \label{word-function-min:welldefined}
There exists a unique vector $\tilde{\alpha}$ satisfying Equation (\ref{minimimisingMWA:defn_alpha}).
For every $\sigma \in \Sigma$, there exists a unique matrix $\tilde{\mu} (\sigma)$ satisfying Equation (\ref{minimimisingMWA:defn_transition}).
\end{lemma}
\begin{proof}
Since the rows of $\tilde{F} \cdot B$ form a basis of $\mathit{RS}(F \cdot B)$, it suffices to prove that $\alpha \cdot B \in \mathit{RS}(F \cdot B)$ and $\mathit{RS}(\tilde{F} \cdot \mu (\sigma) \cdot B) \subseteq \mathit{RS}(F \cdot B)$ for every $\sigma \in \Sigma$.
By  Lemma~\ref{lemma:tree-function-min:rowsubset}, it further suffices to prove that $\alpha \in \mathit{RS}(F)$ and $\mathit{RS}(\tilde{F} \cdot \mu (\sigma)) \subseteq \mathit{RS}(F)$ for every $\sigma \in \Sigma$.

We have $\alpha = \alpha \cdot \mu(\varepsilon) \in \forwardspace = \mathit{RS}(F)$. Let $i \in [m]$. Since   $\tilde{F}_{i} \in \mathit{RS}(F) = \forwardspace$,  it follows from Proposition~\ref{prop:forwardspace_charact} (a) that $(\tilde{F} \cdot \mu (\sigma))_{i} =  \tilde{F}_{i} \cdot \mu (\sigma) \in \forwardspace$ for all $\sigma \in \Sigma$.
\end{proof}
We complete the proof of Proposition \ref{prop-minimal-MWA-abstract} by showing that MWA $\tilde{\A}$ minimises $\A$:

\begin{lemma}
The automaton $\tilde{\A}$ is a minimal MWA equivalent to $\A$.
\end{lemma}
\begin{proof}
We claim that  for every $w \in \Sigma^{*}$,
\begin{align}
\tilde{\alpha} \cdot \tilde{\mu} (w) \cdot \tilde{F} \cdot B = \alpha \cdot \mu(w) \cdot B. \label{MWA:minimality_proof_induction}
\end{align}
Our proof is by induction on $|w|$. For the base case $w = \varepsilon$, we have
\[
\tilde{\alpha} \cdot \tilde{\mu} (\varepsilon) \cdot \tilde{F} \cdot B = \tilde{\alpha} \cdot \tilde{F} \cdot B \mathop{=}^\text{Eq.~(\ref{minimimisingMWA:defn_alpha})} \alpha \cdot B = \alpha \cdot \mu(\varepsilon) \cdot B.
\]
For the induction step, let $l \in \N_0$ and assume that (\ref{MWA:minimality_proof_induction}) holds for every $w \in \Sigma^{l}$.
Take any $w \in \Sigma^{l}$ and $\sigma \in \Sigma$. For every $b \in \backwardspace = \mathit{CS}(B)$ we have by Proposition \ref{prop:backwardspace_charact} (a) that $\mu (\sigma) \cdot b \in \backwardspace$, and thus by the induction hypothesis for $w \in \Sigma^{l}$  it follows
\allowdisplaybreaks \begin{align*}
\tilde{\alpha} \cdot \tilde{\mu} (w \sigma) \cdot \tilde{F} \cdot b
= \tilde{\alpha} \cdot \tilde{\mu} (w) \cdot \tilde{\mu} (\sigma) \cdot \tilde{F} \cdot b
\mathop{=}^\text{Eq.~(\ref{minimimisingMWA:defn_transition})} &\tilde{\alpha} \cdot \tilde{\mu} (w) \cdot \tilde{F} \cdot \mu (\sigma) \cdot b\\
= \hspace{1.1em} &\alpha \cdot \mu(w) \cdot \mu (\sigma) \cdot b = \alpha \cdot \mu(w \sigma) \cdot b
\end{align*}
which completes the proof by induction.

Now for any $w \in \Sigma^{*}$, since $\gamma \in \backwardspace$ we have
\allowdisplaybreaks \begin{align*}
\|\tilde{\A} \| (w) =  \tilde{\alpha} \cdot \tilde{\mu} (w) \cdot \tilde{\gamma} =  \tilde{\alpha} \cdot \tilde{\mu} (w) \cdot \tilde{F} \cdot \gamma \mathop{=}^\text{Eq.~(\ref{MWA:minimality_proof_induction})} \alpha \cdot \mu(w) \cdot \gamma = \|\A \| (w).
\end{align*}
Hence, MWAs $\tilde{\A}$ and $\A$ are equivalent.
Minimality of $\tilde{\A}$ follows from Lemma~\ref{lem-rankFB}.
\end{proof}

We are now ready to prove Proposition~\ref{prop-minimal-MTA-abstract} in its full generality. The proof is split in two lemmas, Lemmas \ref{lem-prop-minimal-MTA-abstract-well-defined} and \ref{lem-prop-minimal-MTA-abstract-equivalent}, which together imply Proposition~\ref{prop-minimal-MTA-abstract}.
First, we show that $\tilde{\A}$ is a well-defined multiplicity tree automaton:

\begin{lemma} \label{lem-prop-minimal-MTA-abstract-well-defined}
For every $\sigma \in \Sigma_{k}$, there exists a unique matrix $\tilde{\mu} (\sigma)$ satisfying Equation (\ref{minimimisingMTA:defn_transition}).
\end{lemma}
\begin{proof}
Since the rows of $\tilde{F} \cdot B$ form a basis of $\mathit{RS}(F \cdot B)$,
 it suffices to prove that 
\[
\mathit{RS}(\tilde{F}^{\otimes k} \cdot \mu (\sigma) \cdot B) \subseteq \mathit{RS}(F \cdot B).
\]
By Lemma~\ref{lemma:tree-function-min:rowsubset}, to do this it suffices to prove that \mbox{$\mathit{RS}(\tilde{F}^{\otimes k} \cdot \mu (\sigma)) \subseteq \mathit{RS}(F)$}.
Let us therefore take an arbitrary row $(\tilde{F}^{\otimes k} \cdot \mu (\sigma))_{(i_1, \ldots, i_k)}$ of $\tilde{F}^{\otimes k} \cdot \mu (\sigma)$, where $(i_1, \ldots, i_k) \in [m]^{k}$. We have
\allowdisplaybreaks \begin{align*}
(\tilde{F}^{\otimes k} \cdot \mu (\sigma))_{(i_1, \ldots, i_k)} = (\tilde{F}^{\otimes k})_{(i_1, \ldots, i_k)} \cdot \mu (\sigma) \mathop{=}^\text{Eq. (\ref{Kronecker_indices})}
(\tilde{F}_{i_1} \otimes \cdots \otimes \tilde{F}_{i_k}) \cdot \mu (\sigma).
\end{align*}
Since $\tilde{F}_{i_1}, \ldots, \tilde{F}_{i_k} \in \mathit{RS}(\tilde{F}) \subseteq \mathit{RS}(F) = \forwardspace$, we have that \mbox{$(\tilde{F}_{i_1} \otimes \cdots \otimes \tilde{F}_{i_k}) \cdot \mu (\sigma) \in \forwardspace$} by Proposition~\ref{prop:forwardspace_charact}~(a). Therefore, $(\tilde{F}^{\otimes k} \cdot \mu (\sigma))_{(i_1, \ldots, i_k)} \in \forwardspace = \mathit{RS}(F)$.
\end{proof}

Next, we show that MTA $\tilde{\A}$ minimises $\A$:

\begin{lemma} \label{lem-prop-minimal-MTA-abstract-equivalent}
The automaton $\tilde{\A}$ is a minimal MTA equivalent to $\A$.
\end{lemma}

\begin{proof}
First we show that for every $t \in T_{\Sigma}$,
\begin{align}
\tilde{\mu} (t) \cdot \tilde{F} \cdot B = \mu(t) \cdot B. \label{lem-prop-minimal-MTA-abstract-equivalent-claim}
\end{align}
Our proof is by induction on $\mathit{height}(t)$. The base case $t = \sigma \in \Sigma_{0}$ follows immediately from Equation (\ref{minimimisingMTA:defn_transition}).
For the induction step, let $h \in \N_0$ and assume that (\ref{lem-prop-minimal-MTA-abstract-equivalent-claim}) holds for every $t \in T_{\Sigma}^{\leq h}$. Take any tree $t \in T_{\Sigma}^{h +1}$. Then $t = \sigma (t_1, \ldots, t_k)$ for some $k \geq 1$, $\sigma \in \Sigma_k$, and $t_1, \ldots, t_k \in T_{\Sigma}^{\leq h}$. Using bilinearity of Kronecker product we get that
\allowdisplaybreaks \begin{align*}
\tilde{\mu} (t) \cdot \tilde{F} \cdot B &= (\tilde{\mu} (t_{1}) \otimes \cdots \otimes \tilde{\mu} (t_{k})) \cdot \tilde{\mu} (\sigma) \cdot \tilde{F} \cdot B\\ &= (\tilde{\mu} (t_{1}) \otimes \cdots \otimes \tilde{\mu} (t_{k})) \cdot  \tilde{F}^{\otimes k} \cdot \mu (\sigma) \cdot B && \text{by Eq.~(\ref{minimimisingMTA:defn_transition})}\\
&= ((\tilde{\mu} (t_{1})\tilde{F}) \otimes \cdots \otimes (\tilde{\mu} (t_{k})\tilde{F})) \cdot \mu (\sigma) \cdot B && \text{by Eq.~(\ref{Kronecker:kmixed})}\\
&= (\tilde{\mu} (t_{1})\tilde{F}) \cdot(I_{n} \otimes (\tilde{\mu} (t_{2})\tilde{F}) \otimes \cdots \otimes (\tilde{\mu} (t_{k})\tilde{F})) \cdot \mu (\sigma)\cdot  B.
\end{align*}
Since $\mathit{RS}(\tilde{F}) \subseteq \forwardspace$, for every $i \in \{2, \ldots, k\}$ it holds that $\tilde{\mu} (t_{i})\tilde{F} \in \forwardspace$. Since $I_{n} = \mu (\Box) \in \forwardspace$, we now have that $(I_{n} \otimes (\tilde{\mu} (t_{2})\tilde{F}) \otimes \cdots \otimes (\tilde{\mu} (t_{k})\tilde{F})) \cdot\mu (\sigma) \cdot B \in \backwardspace$ by Proposition \ref{prop:backwardspace_charact} (a). Thus by the induction hypothesis for $t_{1}  \in T_{\Sigma}^{\leq h}$, we have
\allowdisplaybreaks \begin{align*}
\tilde{\mu} (t) \cdot \tilde{F} \cdot B &= \mu (t_{1}) \cdot(I_{n} \otimes (\tilde{\mu} (t_{2})\tilde{F}) \otimes \cdots \otimes (\tilde{\mu} (t_{k})\tilde{F})) \cdot \mu (\sigma) \cdot B\\
&=  (\mu (t_{1}) \otimes (\tilde{\mu} (t_{2})\tilde{F}) \otimes \cdots \otimes (\tilde{\mu} (t_{k})\tilde{F})) \cdot \mu (\sigma) \cdot B.
\end{align*}
From here we argue inductively as follows: Assume that for some $l \in [k-1]$,
\allowdisplaybreaks \begin{align*}
\tilde{\mu} (t) \cdot \tilde{F} \cdot B =  (\mu (t_{1}) \otimes \cdots \otimes \mu (t_{l}) \otimes (\tilde{\mu} (t_{l+1})\tilde{F}) \otimes \cdots \otimes (\tilde{\mu} (t_{k})\tilde{F})) \cdot \mu (\sigma) \cdot B.
\end{align*}
Then by bilinearity of Kronecker product, we get that $\tilde{\mu} (t) \cdot \tilde{F} \cdot B$ is equal to
\allowdisplaybreaks \begin{align*}
(\tilde{\mu} (t_{l+1})\tilde{F}) \cdot (\mu (t_{1}) \otimes \cdots \otimes \mu (t_{l}) \otimes I_{n} \otimes (\tilde{\mu} (t_{l+2})\tilde{F}) \otimes \cdots \otimes (\tilde{\mu} (t_{k})\tilde{F})) \cdot \mu (\sigma) \cdot B.
\end{align*}
Here $(\mu (t_{1}) \otimes \cdots \otimes \mu (t_{l}) \otimes I_{n} \otimes (\tilde{\mu} (t_{l+2})\tilde{F}) \otimes \cdots \otimes (\tilde{\mu} (t_{k})\tilde{F})) \cdot \mu (\sigma) \cdot B \in \backwardspace$ by the same reasoning as above. The induction hypothesis for $t_{l+1}  \in T_{\Sigma}^{\leq h}$ now implies 
\allowdisplaybreaks \begin{align*}
&\tilde{\mu} (t) \cdot \tilde{F} \cdot B\\ &=  \mu (t_{l+1}) \cdot (\mu (t_{1}) \otimes \cdots \otimes \mu (t_{l}) \otimes I_{n} \otimes (\tilde{\mu} (t_{l+2})\tilde{F}) \otimes \cdots \otimes (\tilde{\mu} (t_{k})\tilde{F})) \cdot \mu (\sigma) \cdot B\\
&=  (\mu (t_{1}) \otimes \cdots \otimes \mu (t_{l}) \otimes \mu (t_{l+1})  \otimes (\tilde{\mu} (t_{l+2})\tilde{F}) \otimes \cdots \otimes (\tilde{\mu} (t_{k})\tilde{F})) \cdot \mu (\sigma) \cdot B.
\end{align*}
Continuing our inductive argument, for $l = k-1$  we get that
\allowdisplaybreaks \begin{align*}
\tilde{\mu} (t) \cdot \tilde{F} \cdot B =  (\mu (t_{1}) \otimes \cdots \otimes \mu (t_{k})) \cdot \mu (\sigma) \cdot B = \mu (t) \cdot B.
\end{align*}
This completes the proof of (\ref{lem-prop-minimal-MTA-abstract-equivalent-claim}) by induction.

Now since $\gamma \in \backwardspace$, for every $t \in T_{\Sigma}$ we have
\allowdisplaybreaks \begin{align*}
\|\tilde{\A}\|(t) = \tilde{\mu} (t)  \cdot \tilde{\gamma} = \tilde{\mu} (t)  \cdot \tilde{F} \cdot \gamma \mathop{=}^\text{Eq. (\ref{lem-prop-minimal-MTA-abstract-equivalent-claim})} \mu (t) \cdot \gamma = \|\A\| (t).
\end{align*}
Hence, MTAs $\tilde{\A}$ and $\A$ are equivalent. Minimality follows from Lemma~\ref{lem-rankFB}.
\end{proof}

By a result of Bozapalidis and Alexandrakis~\cite[Proposition 4]{bozapalidis1989representations},
 all equivalent minimal multiplicity tree automata are equal up to a change of basis.
Thus the MTA~$\tilde{\A}$ is ``canonical''
 in the sense that any minimal MTA equivalent to~$\A$ can be obtained from~$\tilde{\A}$ via a linear transformation:
any $m$-dimensional MTA $\tilde{\A}' = (m, \Sigma, \tilde{\mu}', \tilde{\gamma}')$
 is equivalent to~$\A$ if and only if there exists an invertible matrix $U \in \mathbb{F}^{m \times m}$
 such that $\tilde{\gamma}' = U \cdot \tilde{\gamma}$ and
 \mbox{$\tilde{\mu}'(\sigma) =  U^{\otimes \mathit{rk}(\sigma)} \cdot \tilde{\mu} (\sigma) \cdot U^{-1}$} for every $\sigma \in \Sigma$.

%\begin{theorem}[{\cite{bozapalidis1989representations}}]
%\label{thm:minimal_unique}
%Let $\Sigma$ be a ranked alphabet, $\mathbb{F}$ be a field, and $f \in \mathbb{F}\langle\langle T_{\Sigma} \rangle\rangle$ be a recognisable tree series. Let  $r$ be the rank (over $\mathbb{F}$) of the Hankel matrix of $f$. Suppose $\mathbb{F}$-multiplicity tree automaton $\A_{1} = (r, \Sigma, \mu_{1}, \gamma_{1})$ recognises $f$. Then for any $\mathbb{F}$-multiplicity tree automaton $\A_{2} = (r, \Sigma, \mu_{2}, \gamma_{2})$, $\A_{2}$ recognises $f$ if and only if there exists an invertible matrix $U \in \mathbb{F}^{r \times r}$ such that $\gamma_{2} = U \cdot \gamma_{1}$ and $\mu_{2} (\sigma) =  U^{\otimes \mathit{rk}(\sigma)} \cdot \mu_{1} (\sigma) \cdot U^{-1}$ for every $\sigma \in \Sigma$.
%\end{theorem}

\subsection{Spanning Sets for the Forward and Backward Spaces} \label{sub-forward-and-backward-spaces}

The minimal automaton~$\tilde{\A}$ from Section~\ref{sub-a-minimal-automaton}
 is defined in terms of matrices $F$~and~$B$ whose rows and columns span
 the forward space~$\forwardspace$ and the backward space~$\backwardspace$, respectively.
In fact, the central algorithmic challenge for minimisation lies
 in the efficient computation of such matrices.
In this section we prove a key result, Proposition~\ref{prop-compact-space} below,
 suggesting a way to compute $F$~and~$B$, which we exploit
in Sections \ref{sub-word-min}~and~\ref{sec-decision}.

Propositions~\ref{prop:forwardspace_charact}~and~\ref{prop:backwardspace_charact}
 and their proofs already suggest an efficient algorithm
  for iteratively computing bases of $\forwardspace$~and~$\backwardspace$.
We make this algorithm more explicit and analyse its unit-cost
 complexity in Section~\ref{function_problem:trees}.
The drawback of the resulting algorithm will be the use of ``if-conditionals'':
 the algorithm branches according to whether certain sets of vectors are linearly independent.
Such conditionals are ill-suited for efficient \emph{parallel} algorithms and
 also for many-one reductions.
Thus it cannot be used for an \NC-algorithm in Section~\ref{sub-word-min}
 nor for a reduction to {\acit} in Section~\ref{sec-decision}.

The following proposition exhibits polynomial-size sets of spanning vectors for $\forwardspace$~and~$\backwardspace$,
 which, as we will see later, can be computed efficiently without branching.
The proposition is based on the \emph{product} automaton $\A \times \A$ defined in Section \ref{subsec:MTA}. It defines a sequence $(f(l))_{l \in \N}$ of row vectors and a sequence $(b(l))_{l \in \N}$  of square matrices. Part~(a) states that
 the vector~$f(n)$ and the matrix~$b(n)$ determine matrices $F$ and~$B$,
 whose rows and columns span $\forwardspace$~and~$\backwardspace$, respectively.
Part~(b) gives a recursive characterisation of the sequences $(f(l))_{l \in \N}$  and $(b(l))_{l \in \N}$. This allows for an efficient computation of $f(n)$~and~$b(n)$.
\newcommand{\stmtpropcompactspace}{
Let $\Sigma$ have rank~$r$.
Let $\A \times \A = (n^{2},\Sigma, \mu', \gamma^{\otimes 2})$ be the product of $\A$~by~$\A$.
For every $l \in \N$, define 
\begin{align*}
f(l) := \sum_{t \in T_{\Sigma}^{< l}} \mu' (t) \in \mathbb{F}^{1 \times n^2}
\text{~ and ~} b(l) := \sum_{c \in C_{\Sigma, T_{\Sigma}^{< n}}^{< l}} \mu' (c) \in \mathbb{F}^{n^2 \times n^2}.
\end{align*}
\begin{enumerate}
\renewcommand{\theenumi}{(\alph{enumi})}
\renewcommand{\labelenumi}{\textup{\theenumi}}
\item
Let $F \in \mathbb{F}^{n \times n}$ be the matrix with $F_{i,j} = f(n) \cdot (e_{i} \otimes e_{j})^\top$ for all $i,j \in [n]$.
Let $B \in \mathbb{F}^{n \times n}$ be the matrix with $B_{i,j} = (e_{i} \otimes e_{j}) \cdot b(n) \cdot \gamma^{\otimes 2}$  for all $i,j \in [n]$.
Then, $\mathit{RS}(F) = \forwardspace$ and $\mathit{CS}(B) = \backwardspace$.
\item
We have $f(1) = \sum_{\sigma \in \Sigma_{0}} \mu' (\sigma)$ and $b(1) = I_{n^{2}}$. For all $l \in \N$, it holds that
\allowdisplaybreaks\begin{align*}
f(l+1) & = \sum\limits_{k=0}^{r} f(l)^{\otimes k} \sum\limits_{\sigma \in \Sigma_{k}} \mu'(\sigma), \text{ and}\\
b(l+1) & = I_{n^{2}} + \sum\limits_{k=1}^{r} \sum_{j=1}^k \left( f(n)^{\otimes (j-1)} \otimes b(l) \otimes f(n)^{\otimes (k-j)} \right) \sum\limits_{\sigma \in \Sigma_{k}} \mu^{\prime}(\sigma).
\end{align*}
\end{enumerate}
}
\begin{proposition} \label{prop-compact-space}
\stmtpropcompactspace
\end{proposition}

\begin{proof}
First, we prove that $\mathit{RS}(F) = \forwardspace$ in part~(a).
Let $\hF \in \mathbb{F}^{T_{\Sigma}^{< n} \times [n]}$ be a matrix such that
$\hF_t = \mu(t)$ for every $t \in T_{\Sigma}^{< n}$.
From Proposition~\ref{prop:forwardspace_charact} (b) it follows that $\mathit{RS}(\hF) = \forwardspace$.
By Lemma~\ref{lem-ATA-rowspace} we now have $\mathit{RS}(\hF^\top \hF) = \mathit{RS}(\hF) = \forwardspace$.
Thus in order to prove that $\mathit{RS}(F) = \forwardspace$, it suffices to show that $\hF^\top \hF = F$.
Indeed, using the mixed-product property of Kronecker product, we have for all $i,j \in [n]$:
\allowdisplaybreaks \begin{align*}
(\hF^{\top} \hF)_{i, j} =    (\hF^{\top})_{i} \cdot (\hF)^{j}
& = \ \sum\limits_{t \in T_{\Sigma}^{< n}} \mu(t)_{i} \cdot \mu(t)_{j}\\
& = \ \sum\limits_{t \in T_{\Sigma}^{< n}} (\mu(t) \cdot e_{i}^\top) \otimes (\mu(t) \cdot e_{j}^\top) \\
&  =  \left(\sum\limits_{t \in T_{\Sigma}^{< n}} (\mu(t) \otimes \mu(t))\right) \cdot (e_{i} \otimes e_{j})^\top \\
& \hspace{-0.99em} \mathop{=}^\text{Prop.~\ref{prop:diff_prod_MTA}} \ \left(\sum\limits_{t \in T_{\Sigma}^{< n}} \mu'(t)\right) \cdot (e_{i} \otimes e_{j})^\top 
 = f(n) \cdot (e_{i} \otimes e_{j})^{\top}. 
\end{align*}

Next, we complete the proof of part~(a) by proving that $\mathit{CS}(B) = \backwardspace$.
To avoid notational clutter,  in the following we write 
\[
C := C_{\Sigma, T_{\Sigma}^{< n}}^{< n}.
\]
Define a matrix $\hB \in \mathbb{F}^{[n] \times C}$ such that
$\hB^c = \mu(c) \cdot \gamma$ for all $c \in C$ .
From Proposition~\ref{prop:backwardspace_charact}~(b) it follows that $\mathit{CS}(\hB) = \backwardspace$.
By Lemma~\ref{lem-ATA-rowspace} we now have $\mathit{CS}(\hB \hB^\top) = \mathit{CS}(\hB) = \backwardspace$. Therefore in order to prove that $\mathit{CS}(B) = \backwardspace$, it suffices to show that $\hB \hB^\top = B$.
Indeed, using the mixed-product property of Kronecker product, we have for all $i,j \in [n]$:
\allowdisplaybreaks\begin{align*}
(\hB \cdot \hB^{\top})_{i, j} &=   (\hB)_{i} \cdot (\hB^{\top})^{j}\\
&= \sum_{c \in C} (\mu(c)_{i} \cdot \gamma) \cdot (\mu(c)_{j} \cdot \gamma)\\
&= \sum_{c \in C} (e_{i} \cdot \mu(c) \cdot \gamma) \otimes (e_{j}  \cdot \mu(c) \cdot \gamma) \\
&= \sum_{c \in C} (e_{i} \otimes e_{j}) \cdot (\mu(c) \otimes \mu(c)) \cdot (\gamma \otimes \gamma ) \\
&= (e_{i} \otimes e_{j}) \cdot \left(\sum_{c \in C} (\mu(c) \otimes \mu(c))\right) \cdot (\gamma \otimes \gamma ) \\
&= (e_{i} \otimes e_{j}) \cdot \left(\sum_{c \in C} \mu'(c)\right) \cdot \gamma^{\otimes 2}
    && \text{(by Proposition~\ref{prop:diff_prod_MTA}~(ii))}\\
&= (e_{i} \otimes e_{j}) \cdot b(n) \cdot \gamma^{\otimes 2}
    && \text{(definition of~$b(n)$)} \\
& = B_{i,j}.
\end{align*}

We turn to the proof of part~(b). Here we do not use the fact that we are dealing with a product automaton. We first prove the statement on~$f(l)$.
The equality $f(1) = \sum_{\sigma \in \Sigma_{0}} \mu' (\sigma)$ follows directly from the definition. For all $l \in \N$,
\[
 T_{\Sigma}^{< l+1} = \{\sigma(t_{1}, \ldots ,t_{k}) :
 0 \le k \le r, \ \sigma \in \Sigma_{k}, \ t_{1}, \ldots, t_{k} \in T_{\Sigma}^{< l}\}\;.
\]
Thus, by bilinearity of Kronecker product, it holds that
\allowdisplaybreaks\begin{align*}
f(l+1)
&= \sum\limits_{t \in T_{\Sigma}^{< l+1}} \mu' (t) \\
&= \sum\limits_{k=0}^{r} \sum\limits_{\sigma \in \Sigma_{k}} \sum\limits_{t_{1} \in T_{\Sigma}^{< l}} \cdots \sum\limits_{t_{k} \in T_{\Sigma}^{< l}} \left( \mu'(t_{1}) \otimes \cdots \otimes \mu'(t_{k}) \right) \cdot \mu'(\sigma)\\
&= \sum\limits_{k=0}^{r} \left( \left( \sum_{t_{1} \in T_{\Sigma}^{< l}} \mu' (t_{1}) \right)  \otimes \cdots \otimes \left( \sum_{t_{k} \in T_{\Sigma}^{< l}} \mu' (t_{k}) \right) \right) \cdot \sum\limits_{\sigma \in \Sigma_{k}}  \mu'(\sigma)\\
&= \sum\limits_{k=0}^{r} \left(\sum_{t \in T_{\Sigma}^{< l}} \mu' (t) \right)^{\otimes k} \sum\limits_{\sigma \in \Sigma_{k}} \mu'(\sigma) \\
&= \sum\limits_{k=0}^{r} f(l)^{\otimes k} \sum\limits_{\sigma \in \Sigma_{k}} \mu'(\sigma) \,.
\end{align*}

Finally, we prove the statement on~$b(l)$. The equality $b(1) = I_{n^{2}}$ follows from the definition.
To avoid notational clutter we write $T := T_{\Sigma}^{< n}$ in the following.
Recall that $f(n) = \sum_{t \in T} \mu'(t)$.
We have for all $l \in \N$:
\begin{align*}
 C_{\Sigma,T}^{< l+1} &= \{ \Box\} \cup \left\{\sigma(t_{1}, \ldots, t_{j-1}, c_{j}, t_{j+1}, \ldots ,t_{k}) :   k \in [r], \  j \in [k], \ \sigma \in \Sigma_{k},
                 \mbox{} \right. \\
                      & \qquad\qquad \qquad \quad \left.     c_{j} \in C_{\Sigma, T}^{< l}, \
                      t_{1}, \ldots, t_{j-1}, t_{j+1}, \ldots, t_{k} \in T \right\}\,.
\end{align*}
Thus, using bilinearity of Kronecker product, we get that
\allowdisplaybreaks\begin{align*}
& b(l+1) \\
&= \sum_{c \in C_{\Sigma,T}^{< l+1}} \mu' (c)\\
&= \mu' (\Box) + \sum_{k=1}^{r} \sum_{j=1}^k \sum\limits_{\sigma \in \Sigma_{k}} \sum\limits_{t_{1}, \ldots, t_{j-1} \in T} \sum\limits_{c_{j} \in C_{\Sigma, T}^{< l}} \sum\limits_{t_{j+1}, \ldots, t_{k} \in T} ( \mu'(t_{1}) \otimes \cdots \otimes \mu'(c_{j})  \mbox{} \\
&  \hspace{85.5mm} \mbox{}\otimes \cdots\otimes \mu'(t_{k}) ) \cdot \mu'(\sigma)\\
&= I_{n^{2}} + \sum_{k=1}^{r} \sum_{j=1}^k \Bigg( \Bigg( \sum_{t_{1} \in T} \mu'(t_{1}) \Bigg)  \otimes \cdots  \otimes \Bigg( \sum_{c_{j} \in C_{\Sigma, T}^{< l}} \mu' (c_{j}) \Bigg)\mbox{} \\
&  \hspace{28mm}\otimes \cdots \otimes \Bigg( \sum_{t_{k} \in T} \mu'(t_{k}) \Bigg) \Bigg) \cdot  \sum\limits_{\sigma \in \Sigma_{k}} \mu'(\sigma)\\
&= I_{n^{2}} + \sum_{k=1}^{r} \sum_{j=1}^k \left( f(n)^{\otimes (j-1)} \otimes b(l) \otimes f(n)^{\otimes (k-j)} \right)
    \sum\limits_{\sigma \in \Sigma_{k}} \mu'(\sigma)\,. %\tag*{\qed}
\end{align*}
This completes the proof.
\end{proof}
Loosely speaking, Proposition~\ref{prop-compact-space} says that
 the sum over a small subset of the forward space of the product automaton encodes
 a spanning set of the whole forward space of the original automaton,
 and similarly for the backward space.

\section{Minimisation Algorithms} \label{sec-function}

In this section we devise algorithms for minimising a given multiplicity automaton: Section~\ref{function_problem:trees} considers general MTAs,
while Section~\ref{sub-word-min} considers MWAs.
For the sake of a complexity analysis in standard models, we fix the field $\mathbb{F} = \mathbb{Q}$.

\subsection{Minimisation of Multiplicity Tree Automata}\label{function_problem:trees}
In this subsection we describe
an implementation of the algorithm
implicit in Section~\ref{sub-a-minimal-automaton},
 and analyse the number of operations. %, assuming unit-cost arithmetic.
We consider a multiplicity tree automaton $\A =(n, \Sigma, \mu, \gamma)$. We denote by~$r$ the rank of~$\Sigma$.
The algorithm has three steps, as follows:

\subsubsection{Step 1 ``Forward''.}
The first step is to compute a matrix~$F$ such that $\mathit{RS}(F) = \forwardspace$.
Seidl~\cite{seidlfull} outlines a saturation-based algorithm for this, 
 and proves that the algorithm takes polynomial time assuming unit-cost arithmetic.
Based on Proposition~\ref{prop:forwardspace_charact}~(a) we now give in Table~\ref{algorithm:F} an explicit version of Seidl's algorithm.  

\begin{table}[!h]
\begin{center}
\begin{tabular}{l}
\multicolumn{1}{c}{\parbox{0.99\textwidth}{
\textbf{Input:} $\mathbb{Q}$-multiplicity tree automaton $(n, \Sigma, \mu, \gamma)$\\
\textbf{Output:} matrix $F$ whose rows form a basis of the forward space $\forwardspace$\\
 \hspace*{2em} $i:=0$, $j:=0$\\
 \hspace*{2em} while $i \le j$ do \\
 \hspace*{4em} forall $\sigma \in \Sigma$ do\\
 \hspace*{6em} forall $(l_{1}, \ldots, l_{\mathit{rk}(\sigma)}) \in [i]^{\mathit{rk}(\sigma)} \setminus{[i-1]^{\mathit{rk}(\sigma)}} $ do\\
 \hspace*{8em} $v := (F_{l_{1}} \otimes \cdots \otimes F_{l_{\mathit{rk}(\sigma)}}) \cdot \mu (\sigma)$\\
 \hspace*{8em} if $v \not\in \langle F_{1}, \ldots, F_{j} \rangle$\\
 \hspace*{10em} $j := j+1$\\
 \hspace*{10em} $F_{j} := v$\\
 \hspace*{4em} $i := i+1$ \\
\hspace*{2em} \textbf{return} matrix $F \in \mathbb{Q}^{j \times n}$
}} \\
\end{tabular}
\end{center}
\caption{Algorithm for computing a matrix $F$}
\label{algorithm:F}
\end{table}

Our algorithm satisfies the following properties:

\begin{lemma} \label{lem-function-forward}
The algorithm in Table~\ref{algorithm:F} returns a matrix~$F \in \mathbb{Q}^{\nF \times n}$ whose rows form a basis of the forward space~$\forwardspace$. Each row of $F$ equals~$\mu(t)$ for some tree $t \in T_\Sigma^{<n}$. The algorithm executes $O\left(\sum_{k=0}^r |\Sigma_k| \cdot n^{2 k + 1}\right)$ operations.
\end{lemma}
\begin{proof}
The fact that the rows of $F$ span~$\forwardspace$ follows from Proposition~\ref{prop:forwardspace_charact} (a). Moreover, it is clear from the algorithm that the rows of $F$ are linearly independent. 

A straightforward induction shows that for each row index $j \ge 1$, the row~$F_j$ equals~$\mu(t)$ for some tree $t \in T_{\Sigma}^{< j}$. %Let  $\nF$ be the number of rows in~$F$. 
The returned  matrix $F \in \mathbb{Q}^{\nF \times n}$ has full row rank, and therefore $\nF \le n$. Hence, each row of $F$ equals~$\mu(t)$ for some tree $t \in T_\Sigma^{<n}$.

It remains to analyse the number of operations. Let us consider an iteration of the innermost ``for'' loop. The computation of $F_{l_{1}} \otimes \cdots \otimes F_{l_{\mathit{rk}(\sigma)}}$ requires $O(n^{\mathit{rk}(\sigma)})$ operations (by iteratively computing partial products). The vector
\[
v = (F_{l_{1}} \otimes \cdots \otimes F_{l_{\mathit{rk}(\sigma)}}) \cdot \mu (\sigma)
\] is the product of a $1 \times n^{\mathit{rk}(\sigma)}$ vector with an
 $n^{\mathit{rk}(\sigma)} \times n$ matrix.
Thus, computing~$v$ takes $O(n^{\mathit{rk}(\sigma)+1})$ operations.
For the purpose of checking membership of~$v$ in the vector space $\forwardspace' := \langle F_{1}, \ldots, F_{j} \rangle$ it is useful to maintain a matrix~$F'$, which is upper triangular (up to a permutation of its columns) and whose rows form a basis of~$\forwardspace'$. To check whether $v \in \forwardspace'$ we compute a vector~$v'$ as the result of performing a Gaussian elimination of~$v$ against~$F'$,
 which requires $O(j \cdot n)$ operations.
If this membership test fails, we extend the matrix~$F'$ at the bottom by row~$v'$. This preserves the upper-triangular shape of~$F'$.
Thus, an iteration of the innermost ``for'' loop takes $O(n^{\mathit{rk}(\sigma)+1})$ operations. For every $\sigma \in \Sigma$, this ``for'' loop is executed $O(n^{\mathit{rk}(\sigma)})$ times.
Therefore, the algorithm executes %$O(|\Sigma| \cdot n^{2 r + 1})$
$O\left(\sum_{k=0}^r |\Sigma_k| \cdot n^{2 k + 1}\right)$ operations.
\end{proof}

\subsubsection{Step 2 ``Backward''.}
The next step suggested in Section~\ref{sub-a-minimal-automaton} is to compute a matrix~$B$ such that $\mathit{CS}(B) = \backwardspace$. By  Lemma~\ref{lem-function-forward}, each row of the matrix~$F$ computed by the algorithm in Table~\ref{algorithm:F}  equals $\mu(t)$ for some tree $t \in T_{\Sigma}^{<n}$. Let $S$ denote the set of those trees. Since $\mathit{RS}(F) = \forwardspace$, set $\{ \mu(t) : t \in S \}$ spans $\forwardspace$. Thus by Proposition~\ref{prop:backwardspace_charact}~(a), $\backwardspace$ is the smallest vector space $V \subseteq \mathbb{Q}^{n}$ such that $\gamma \in V$ and $M \cdot v \in V$ for all  $M \in \mathcal{M} := \{ \mu(c) : c \in C^1_{\Sigma,S} \}$ and $v \in V$.
Tzeng~\cite{Tzeng92} shows, for an arbitrary column vector $\gamma \in \mathbb{Q}^n$
 and an arbitrary finite set of matrices $\mathcal{M} \subseteq \mathbb{Q}^{n \times n}$,
  how to compute a basis of~$V$ in time $O (|\mathcal{M}| \cdot n^{4})$.
This can be improved to $O (|\mathcal{M}| \cdot n^{3})$ (see,~\eg,~\cite{CortesMohriRastogi}).
This leads to the following lemma:

\begin{lemma} \label{lem-function-backward}
Given the matrix~$F \in \mathbb{Q}^{\nF \times n}$ which is the output of the algorithm in Table~\ref{algorithm:F},
 a matrix~$B$ whose columns form a basis of the backward space~$\backwardspace$ can be computed with
%$O(|\Sigma| \cdot r \cdot (n^{2 r} + n^{r+2}))$
$O\left(\sum_{k=1}^r |\Sigma_k| \cdot (k n^{2 k} + k n^{k+2})\right)$
operations.
\end{lemma}
\begin{proof}
Consider the computation of an arbitrary $M \in \mathcal{M} := \{ \mu(c) : c \in C^1_{\Sigma,S} \}$. We have:
 \begin{align}
  M & = G \cdot \mu(\sigma)\;, \quad \text{where}
   \label{eq-function-backward-M} \\%(F_{l_{1}} \otimes \cdots \otimes F_{l_{i-1}} \otimes I_{n} \otimes F_{l_{i+1}} \otimes \cdots \otimes F_{l_{\mathit{rk}(\sigma)}}) \cdot \mu (\sigma)\;,
  G & = F_{l_{1}} \otimes \cdots \otimes F_{l_{i-1}} \otimes I_{n} \otimes F_{l_{i+1}} \otimes \cdots \otimes F_{l_{\mathit{rk}(\sigma)}} \in \mathbb{Q}^{n \times n^{\mathit{rk}(\sigma)}}\; \label{eq-function-backward}
  \end{align}
is such that $\sigma \in \Sigma \setminus{\Sigma_{0}}$, $i \in [\mathit{rk}(\sigma)]$,
$l_{1}, \ldots, l_{i-1}, l_{i+1}, \ldots, l_{\mathit{rk} (\sigma)} \in [\nF]$.

Exploiting the sparsity pattern in the matrix $G$ as in~\eqref{eq-function-backward}, the computation of the non-zero entries of~$G$ takes $O(n^{\mathit{rk}(\sigma)})$ operations.
Exploiting sparsity again, the computation of matrix~$M$ as in~\eqref{eq-function-backward-M} then takes $O(n^{\mathit{rk}(\sigma)+1})$ operations.
Since~$\nF \le n$, it follows from \eqref{eq-function-backward-M} and \eqref{eq-function-backward} that
 \[
  |\mathcal{M}| \in O\left(\sum_{k=1}^r |\Sigma_k| \cdot k \cdot n^{k-1}\right)\,.
 \]
Thus, the number of operations required to compute~$\mathcal{M}$ is $O\left(\sum_{k=1}^r |\Sigma_k| \cdot k \cdot n^{2 k}\right)$.
Given~$\mathcal{M}$, computing a basis of~$\backwardspace$ takes
\[
O(|\mathcal{M}| \cdot n^3) = O\left(\sum_{k=1}^r |\Sigma_k| \cdot k \cdot n^{k-1} \cdot n^3\right)
\]
operations, using, \eg, the method from~\cite{CortesMohriRastogi} that was mentioned above. Therefore, the total operation count for computing a matrix~$B$ is
 $O\left(\sum_{k=1}^r |\Sigma_k| \cdot (k n^{2 k} + k n^{k+2})\right)$.
\end{proof} 

\subsubsection{Step 3 ``Solve''.}
The final step suggested in Section~\ref{sub-a-minimal-automaton} has two substeps. The first substep is to compute a matrix~$\tilde{F} \in \mathbb{Q}^{m \times n}$ with $m = \mathit{rank}(F \cdot B)$ and $\mathit{RS}(\tilde{F} \cdot B) = \mathit{RS}(F \cdot B)$.
Such a matrix~$\tilde{F}$ can be computed from~$F$ by going through the rows of $F$ one by one and including only those rows
 that are linearly independent of the previous rows when multiplied by~$B$.
This can be done in time~$O(n^3)$, \eg, by transforming the matrix $F \cdot B$ into a triangular form using Gaussian elimination.

The second substep is to compute the minimal MTA~$\tilde{\A} = (m, \Sigma, \tilde{\mu}, \tilde{\gamma})$.
The vector $\tilde{\gamma} = \tilde{F} \cdot \gamma$ is easy to compute.
Solving Equation~\eqref{minimimisingMTA:defn_transition} for each $\tilde{\mu}(\sigma)$
 can be done via Gaussian elimination in time~$O(n^3)$; however, the bottleneck is the computation of $\tilde{F}^{\otimes k} \cdot \mu (\sigma)$ for every $\sigma \in \Sigma_{k}$,
  which takes
\[ 
O\left(\sum_{k=0}^r |\Sigma_k| \cdot n^k \cdot n^k \cdot n\right) =
   O\left(\sum_{k=0}^r |\Sigma_k| \cdot n^{2 k+1}\right) 
\]
operations. Putting together the results of this subsection, we get:
\begin{theorem} \label{thm-function-tree}
There is an algorithm that transforms a given $\mathbb{Q}$-MTA $\A = (n, \Sigma, \mu, \gamma)$ into an equivalent minimal $\mathbb{Q}$-MTA.
Assuming unit-cost arithmetic,
the algorithm takes time
\[
O\left(\sum_{k=0}^r |\Sigma_k| \cdot (n^{2 k + 1} + k n^{2k} + k n^{k+2})\right),
\]
which is $O\left(|\A|^2 \cdot r\right)$. \qed
\end{theorem}

\subsection{Minimisation of Multiplicity Word Automata in \NC} \label{sub-word-min}

In this subsection, we consider the problem of minimising a given $\mathbb{Q}$-multiplicity word automaton. We prove the following result:
\begin{theorem} \label{thm-word-min}
 There is an \NC\ algorithm that transforms a given $\mathbb{Q}$-MWA
  into an equivalent minimal $\mathbb{Q}$-MWA.
 In particular, given a $\mathbb{Q}$-MWA and a number $d \in \N_{0}$,
  one can decide in~\NC\ whether there exists an equivalent $\mathbb{Q}$-MWA of dimension at most~$d$.
\end{theorem}
Theorem~\ref{thm-word-min} improves on two results of~\cite{KieferMOWW13}.
First, \cite[Theorem~4.2]{KieferMOWW13} states that
 deciding whether a $\mathbb{Q}$-MWA is minimal is in~\NC.
Second, \cite[Theorem~4.5]{KieferMOWW13} states the same thing as our Theorem~\ref{thm-word-min}, but with \NC\ replaced with \emph{randomised}~\NC.
\begin{proof}[Proof of Theorem~\ref{thm-word-min}]
The algorithm relies on Propositions \ref{prop-minimal-MWA-abstract}~and~\ref{prop-compact-space}.
Let the given $\mathbb{Q}$-MWA be $\A = (n, \Sigma, \mu, \alpha, \gamma)$.
In the notation of Proposition~\ref{prop-compact-space}, we have
 for all $l \in \N$ that 
\[
b(l+1) = I_{n^{2}} + b(l) \cdot \sum_{\sigma \in \Sigma} \mu'(\sigma). 
\]
From here one can easily show, using an induction on $l$, that for all $l \in \N$:
\[
b(l) = \sum_{k=0}^{l-1} \left( \sum_{\sigma \in \Sigma} \mu'(\sigma) \right)^k.
\]
It follows for the matrix $B \in \mathbb{Q}^{n \times n}$ from Proposition~\ref{prop-compact-space} that for all $i, j \in [n]$:
\[
 B_{i,j} = (e_{i} \otimes e_{j}) \cdot b(n) \cdot \gamma^{\otimes 2} = (e_{i} \otimes e_{j}) \cdot \left( \sum_{k=0}^{n-1} \Big( \sum_{\sigma \in \Sigma} \mu'(\sigma) \Big)^k \right)\cdot \gamma^{\otimes 2}.
\]
 Note that, since $\A$ is an MWA, we have $f(l) = b(l)$ for all $l \in \N$. We now have for the matrix $F \in \mathbb{Q}^{n \times n}$ from Proposition~\ref{prop-compact-space} and  all $i, j \in [n]$:
\[
 F_{i,j} = \alpha^{\otimes 2} \cdot \left( \sum_{k=0}^{n-1} \Big( \sum_{\sigma \in \Sigma} \mu'(\sigma) \Big)^k \right)
          \cdot (e_{i} \otimes e_{j})^\top.
\]
The matrices $F$ and $B$ can be computed in~\NC\ since sums and matrix powers can be computed in~\NC~\cite{Cook85}.
%With respect to the definition of~$\tilde{\A}$ from Proposition~\ref{prop-minimal-MTA-abstract}, we take $F := F$ and $B := B$.
Next we show how to compute in~\NC\ the matrix~$\tilde{F}$,
 which is needed to compute the minimal $\mathbb{Q}$-MWA~$\tilde{\A}$
 from Section~\ref{sub-a-minimal-automaton}.
Our \NC\ algorithm includes the $i^\text{th}$ row of~$F$ (\ie, $F_i$) in~$\tilde{F}$
 if and only if 
\[
\mathit{rank}(F_{[i],[n]} \cdot B) > \mathit{rank}(F_{[i-1],[n]} \cdot B).
\]
This can be done in~\NC\ since the rank of a matrix can be determined in~\NC~\cite{IbarraMR80}.
It remains to compute $\tilde{\gamma} := \tilde{F} \gamma$ and solve Equations
 \eqref{minimimisingMWA:defn_alpha} and~\eqref{minimimisingMWA:defn_transition}
 for $\tilde{\alpha}$ and $\tilde{\mu}(\sigma)$, respectively.
Both are easily done in~\NC.
\end{proof}

\section{Decision Problem} \label{sec-decision}

In this section we characterise the complexity of the following decision problem:
Given a \mbox{$\mathbb{Q}$-MTA} and a number $d \in \N_0$, the \emph{minimisation} problem asks
 whether there is an equivalent $\mathbb{Q}$-MTA of dimension at most~$d$.
We show, in Theorem~\ref{thm-decision} below, that this problem is interreducible
 with the arithmetic circuit identity testing (\acit) problem.

The latter problem can be defined as follows.
An \emph{arithmetic circuit} is a finite directed acyclic vertex-labelled multigraph whose vertices,
 called \emph{gates}, have indegree $0$~or~$2$.
Vertices of indegree $0$, called \emph{input gates}, are labelled with a nonnegative integer or a variable from the set $\{ x_{i} : i \in \mathbb{N}\}$.
Vertices of indegree $2$ are labelled with one of the arithmetic operations
 $\mathord{+}$, $\mathord{\times}$, or~$\mathord{-}$.
One can associate, in a straightforward inductive way, each gate with the polynomial it computes.
The \emph{arithmetic circuit identity testing (\acit)} problem asks,
 given an arithmetic circuit and a gate,
 whether the polynomial computed by the gate is equal to the zero polynomial.
We show:
\begin{theorem} \label{thm-decision}
Minimisation is logspace interreducible with~\acit.
\end{theorem}
We consider the lower and the upper bound separately.

\subsection{Lower Bound.} %\label{sub-decision-lower}
Given a $\mathbb{Q}$-MTA~$\A$, the \emph{zeroness} problem asks whether $\|\A\|(t) = 0$ for all trees~$t$.
Observe that $\|\A\|(t) = 0$ for all trees~$t$ if and only if there exists an equivalent automaton of dimension~$0$.
Therefore, zeroness is a special case of minimisation.

We observe that there is a logspace reduction from~\acit\ to zeroness.
Indeed, it is shown in~\cite{DBLP:conf/mfcs/MarusicW14} that the \emph{equivalence} problem for {$\mathbb{Q}$-MTAs} is logspace equivalent to~\acit.
This problems asks, given two $\mathbb{Q}$-MTAs $\A_1$ and $\A_2$, whether $\|\A_1\|(t) = \|\A_2\|(t)$ for all trees~$t$. By Proposition~\ref{prop:diff_prod_MTA}, one can reduce this problem to zeroness in logarithmic space.
This implies \acit-hardness of minimisation.

\subsection{Upper Bound.} %\label{sub-decision-upper}
We prove:
\begin{proposition} \label{prop-decision-upper}
There is a logspace reduction from minimisation to~\acit.
\end{proposition}
\begin{proof}
Let $\A = (n, \Sigma, \mu, \gamma)$ be the given $\mathbb{Q}$-MTA, and let $d \in \N_0$ be the given number. In our reduction to~\acit, we allow input gates with rational labels as well as division gates. Rational numbers and division gates can be eliminated in a standard way by constructing separate gates for the numerators and denominators of the rational numbers computed by the original gates.

By Lemma~\ref{lem-rankFB}, the dimension of a minimal MTA equivalent to $\A$ is $m := \mathit{rank}(F \cdot B)$ where $F$ and  $B$ are matrices such that $\mathit{RS}(F) = \forwardspace$ and $\mathit{CS}(B) = \backwardspace$.
Therefore, we have $m \le d$ if and only if $\mathit{rank}(F \cdot B) \le d$.
The recursive characterisation of $F$~and~$B$ from Proposition~\ref{prop-compact-space}
 allows us to compute in logarithmic space an arithmetic circuit for~$F \cdot B$.
Thus, the result follows from Lemma~\ref{lem-bens-lemma} below.
\end{proof}

The following lemma follows easily from the well-known \NC\ procedure
for computing matrix rank~\cite{Csanky76}.

\begin{lemma} \label{lem-bens-lemma}
Let $M \in \mathbb{Q}^{m \times n}$ and $d \in \N_0$.
The problem of deciding whether $\mathit{rank}(M) \le d$ is logspace reducible to~\acit.
\end{lemma}
\begin{proof}
  By the rank-nullity theorem, we have that $\mathit{rank}(M)\leq d$
  if and only if $\mathit{dim} (\mathit{ker}(M)) \geq n-d$.
  Since $\mathit{ker}(M) = \mathit{ker}(M^{\top} M)$, this is equivalent to
  $\mathit{dim} (\mathit{ker}(M^{\top} M)) \geq n-d$.  The matrix~$M^{\top} M$
  is Hermitian, therefore $\mathit{dim} (\mathit{ker}(M^{\top} M)) \geq
  n-d$ if and only if the $n-d$ lowest-order coefficients of the
  characteristic polynomial of $M^{\top} M$ are all zero~\cite{IbarraMR80}.
  But these coefficients are representable by arithmetic circuits with
  inputs from~$M$ (see \cite{Csanky76}).  
\end{proof}

We emphasise that our reduction to~\acit\ is a many-one reduction, thanks to Proposition~\ref{prop-compact-space}:
 our reduction computes only a single instance of~\acit; there are no if-conditionals.

\section{Minimal Consistent Multiplicity Automaton} \label{sec-min-consis}

Let $\mathbb{F}$ be an arbitrary field. A natural computational
problem is to compute an $\mathbb{F}$-MWA $\A$ of minimal dimension
that is consistent with a given finite set of $\mathbb{F}$-weighted
words $S=\{(w_1,r_1),\ldots,(w_m,r_m)\}$, where $w_i \in \Sigma^*$ and
$r_i\in \mathbb{F}$ for every $i \in [m]$. Here \emph{consistency}
means that $\|\A\| (w_i) = r_i$ for every $i \in [m]$.

The main result of this section concerns the computability of the
above consistency problem for the field of rational numbers.  More
specifically, we consider a decision version of this problem, which we
call \emph{minimal consistency problem}, which asks whether there
exists a $\mathbb{Q}$-MWA consistent with a set of input-output
behaviours $S \subseteq \Sigma^*\times \mathbb{Q}$ and that has
dimension at most some nonnegative integer bound $n$.

We show that the minimal consistency problem is logspace equivalent to
the problem of deciding the truth of existential first-order sentences
over the structure~$(\mathbb{Q},+,\cdot,0,1)$.  The decidability of
the latter is a longstanding open
problem~\cite{Pheidas:Hiblerts-tenth}.  This should be compared with
the result that the problem of finding the smallest deterministic
finite automaton consistent with a set of accepted or rejected words
is \NP-complete~\cite{Gold78}.
 
The reduction of the minimal consistency problem to the decision
problem for existential first-order sentences over the structure~$(\mathbb{Q},+,\cdot,0,1)$ is immediate.  The idea is to represent a $\mathbb{Q}$-MWA
$\A=(n,\Sigma,\mu,\alpha,\gamma)$ ``symbolically'' by introducing
separate variables for each entry of the initial weight
vector~$\alpha$, final weight vector~$\gamma$, and each transition
matrix~$\mu(\sigma)$, $\sigma \in \Sigma$.  Then, the consistency of
automaton~$\A$ with a given finite sample
$S\subseteq \Sigma^*\times \mathbb{Q}$ can directly be written as an
existential sentence.

We note in passing that the minimal consistency problem for weighted
word and tree automata over the field $\mathbb{R}$ is in like manner
reducible to the problem of deciding the truth of existential
first-order sentences over the structure~$(\mathbb{R},+,\cdot,0,1)$,
which is well known to be decidable in
\PSPACE~\cite{Canny88}.\footnote{To consider this problem within the
  conventional Turing model, we assume that the set $S$ of
  input-output behaviours is still a subset of
  $\Sigma^*\times \mathbb{Q}$.  Of course, the dimension of the
  smallest MWA consistent with a given finite set of behaviours $S$
  depends on the weight field of the output automaton.}

Conversely, we reduce the decision problem for existential first-order
sentences over the structure~$(\mathbb{Q},+,\cdot,0,1)$ to the minimal
consistency problem for $\mathbb{Q}$-MWA.  In fact it suffices to
consider sentences in the restricted form
\begin{gather}
\exists x_1 \cdots \exists x_n \bigwedge_{i=1}^m f_i(x_1,\ldots,x_n)=0 \, ,
\label{eq:conj}
\end{gather}
where
$f_i(x_1,\ldots,x_n) = \sum_{j=1}^{l_{i}} c_{i,j}
x_1^{k_{i,j,1}}\cdots x_n^{k_{i,j,n}}$
is a polynomial with rational coefficients.  We can make this
simplification without loss of generality since a disjunction of
atomic formulas $f=0 \vee g=0$, where $f$ and $g$ are polynomials, can
be rewritten to
 \[\exists
x\,(x^2-x=0 \wedge x \cdot f=0 \wedge (1-x) \cdot g=0)
\, . \]  Moreover, the negation of an atomic formula $f\neq 0$ is equivalent
to $\exists x\, (x \cdot f=1)$.

Define an alphabet 
\[
\Sigma := \{s,t\} \cup \{ \#_i,\bar{c}_{i,j},\bar{x}_k : i \in [m], j \in [l_i], k \in [n]\},
\]
including symbols $\bar{c}_{i,j}$ and $\bar{x}_k$ for each coefficient
$c_{i,j}$ and variable $x_k$, respectively.  Over the alphabet $\Sigma$ we
consider the 3-dimensional $\mathbb{Q}$-MWA $\A$, depicted in
Figure~\ref{fig:hankel}~(b).  The transitions in this automaton are
annotated by label-weight pairs in $\Sigma\times\mathbb{Q}$.  Recall
that the weights $c_{i,j}$ are coefficients of the polynomial $f_i$.
For each $k\in [n]$, the weight $a_k$ is a fixed but arbitrary element
of $\mathbb{Q}$.

\begin{figure}
\begin{center}
\begin{tabular}{c@{\hspace{20mm}}c}
$\begin{array}{c|ccc}
                   & st     & \;t\;     & \;\varepsilon \;\\ 
                   \hline
\varepsilon        &  1     & 0     & 0\\
s                  &  0     & 1     & 0\\
st                 &  0     & 0     & 1\\ \hline
\#_i                 &  1     & 1     & 0\\
s\#_i                &  0     & 0     & 1\\
st\#_i               &  0     & 0     & 1\\ \hline
\bar{c}_{i,j}     &  1     & 0     & 0\\
s\bar{c}_{i,j}    &  0     & c_{i,j}   & 0\\
st\bar{c}_{i,j}   &  0     & 0     & 1\\ \hline
\bar{x}_k     &  1     & 0     & 0\\
s\bar{x}_k    &  0     & a_k    & 0\\
st\bar{x}_k   &  0     & 0     & 1\\  \hline
 t             &  0     &  0     &0\\
 stt           & 0     & 0       &0\\
 ss            & 0     & 0       &0\\
 sts           & 0     & 0       &0\\ \hline
\end{array}$
&
\begin{tikzpicture}[baseline=(q0),>=stealth',shorten >=1pt,
auto,node distance=2.3cm]
  \node[state] (q0)      {};
  \node[state]         (q1) [right of=q0]  {};
  \node[state]         (q2) [right of=q1] {$1$};

  \path[->] (-1,0) edge (q0);
  \path[->]   (q0)  edge [loop above] node {($\#_i$,$1$)} (q0);
  \path[->]  (q0)  edge [loop below] node[text width=1cm,align=center] {($\bar{c}_{i,j}$,$1$)\\ ($\bar{x}_{k}$,$1$)} (q0);
  \path[->]  (q0)  edge  node {($\#_i$,$1$)} node [swap] {($s$,$1$)} (q1);
   \path[->]      (q1)  edge [loop below] node[text width=1cm,align=center] {($\bar{c}_{i,j}$,$c_{i,j}$)\\ ($\bar{x}_{k}$,$a_k$)} (q1);
  \path[->] (q1) edge node {($\#_i$,1)} node [swap] {($t$,1)} (q2);
  \path   (q2) edge [loop above] node {($\#_i$,1)} (q2);
   \path[->] (q2)  edge [loop below] node[text width=1cm,align=center] {($\bar{c}_{i,j}$,1)\\ ($\bar{x}_{k}$,1)} (q2);
\end{tikzpicture}
\\[35mm]
(a) & (b)
\end{tabular}
\end{center}
\caption{The left figure~(a) shows a Hankel-matrix fragment
  $\tilde{H}$, where $i\in [m]$, $j\in [l_{i}]$, $k\in [n]$.  The
  right figure~(b) shows a graph representation of the automaton
  $\A$.}
\label{fig:hankel}
\end{figure}
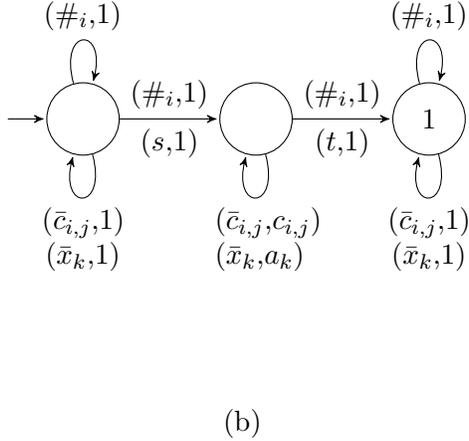

Define $X,Y\subseteq \Sigma^*$ by $X = \{\varepsilon,s,st\}$ and $Y =
\{st,t,\varepsilon\}$. Consider the fragment $\tilde{H} :=H_{X\cup
  X\Sigma, Y}$, shown in Figure~\ref{fig:hankel}~(a), of the Hankel matrix~$H$ of~$\|\A\|$.  We know from Theorem~\ref{thm:Hankel} that $\mathit{rank}(H) \le 3$. Since $\mathit{rank}(H_{X, Y}) = 3$, we have $\mathit{rank}(H_{X, Y}) = \mathit{rank}(H) = 3$. Now, from Remark~\ref{rem:Hankelword} it follows that any
$3$-dimensional $\mathbb{Q}$-MWA $\A'$ that is consistent with $H_{X,Y}$ and $H_{X\Sigma,Y}$ (\ie, consistent with $\tilde{H}$) is equivalent to $\A$. %, i.e., $\|\A\| \equiv \|\A'\|$.

Now for every $i \in [m]$, we encode polynomial $f_i$ by the word
\[ w_i:=\#_i\bar{c}_{i,1}{\bar{x}_1^{k_{i,1,1}}}\cdots
{\bar{x}_n^{k_{i,1,n}}}\#_i \cdots \#_i
\bar{c}_{i,l_{i}}{\bar{x}_1^{k_{i,l_{i},1}}}\cdots
{\bar{x}_n^{k_{i,l_{i},n}}}\#_i 
\]
over alphabet $\Sigma$. Note that $w_i$ comprises $l_i$ `blocks' of symbols, corresponding to the $l_i$
monomials in $f_i$, with each block enclosed by two $\#_i$ symbols.
From the definition of~$w_i$ it follows that $\|\A\|(w_i) =
f_i(a_1,\ldots,a_n)$; the details are given below in the proof of Proposition~\ref{prop:consistency}.

We define a set of weighted words $S\subseteq \Sigma^*\times \mathbb{Q}$
 as $S := S_{1} \cup S_{2}$, where $S_{1}$ is the set of all pairs $(uv, \tilde{H}_{u, v})$ with $u \in X\cup
  X\Sigma$, $v \in Y$, and $uv \not\in \{ s\bar{x}_kt : k \in [n]
\}$, and $S_{2} := \{(w_i,0) : i \in [m]\}$. That is, $S_{1}$ specifies all entries in the matrix $\tilde{H}$ except those that are in row $s\bar{x}_k$ and column $t$. 
%\begin{align*}
%S = \{(uv, \tilde{H}_{u, v}) : u \in X\cup  X\Sigma, v \in Y, uv \not\in \{ s\bar{x}_kt : k \in [n] \}\} \cup \{(w_i,0) : i \in [m]\}
%\end{align*} 

Any 3-dimensional $\mathbb{Q}$-MWA $\A'$ consistent with $S_{1}$ is equivalent to an automaton of the form $\A$ for some $a_1,\ldots,a_n \in \mathbb{Q}$. If $\A'$ is moreover consistent with $S_{2}$, then $f_i(a_1,\ldots,a_n) = 0$ for every $i \in [m]$. From this observation we have the following
proposition:
\begin{proposition} \label{prop:consistency}
The sample $S$ is consistent with a 3-dimensional $\mathbb{Q}$-MWA if and only if the sentence  (\ref{eq:conj}) is true in $(\mathbb{Q},+,\cdot,0,1)$.
\end{proposition}
\begin{proof}
  We have already noted that any 3-dimensional $\mathbb{Q}$-MWA consistent with $S$  must be equivalent to an automaton of the form $\A$ in Figure~\ref{fig:hankel}~(b) for some $a_1,\ldots,a_n \in \mathbb{Q}$.  However, such an automaton is
  consistent with $S$ if and only if it assigns weight $0$ to each
  word $w_i$, $i \in [m]$.  Now, we claim that this is the case if and only if
  $(a_1,\ldots,a_n)$ is a root of $f_{i}$ for every $i \in [m]$, where $a_k$ is the weight of the $\bar{x}_k$-labelled self-loop in the middle state, for
  every $k\in[n]$.

 For every $i \in [m]$, the word $w_i$ has $l_i$ different accepting runs in $\A$, one for
  each monomial in $f_i$.  The $j^\text{th}$ such run, in which the
  block $\bar{c}_{i,j}\bar{x}_1^{k_{i,j,1}}\cdots
  \bar{x}_{n}^{k_{i,j,n}}$ is read in the middle state, has weight
  $c_{i,j}a_1^{k_{i,j,1}}\cdots a_n^{k_{i,j,n}}$, \ie, the value of
  monomial $c_{i,j}x_1^{k_{i,j,1}}\cdots x_n^{k_{i,j,n}}$ evaluated at
  $(a_1,\ldots,a_n)$.  Thus $\|\A\|(w_i) = f_i(a_1,\ldots,a_n)$.
\end{proof}

From Proposition~\ref{prop:consistency} we derive the main result of
this section:
\begin{theorem} \label{thm-min-consis} The minimal consistency problem
  for $\mathbb{Q}$-MWAs is logspace equivalent to the decision problem
  for existential first-order sentences
  over~$(\mathbb{Q},+,\cdot,0,1)$. \qed 
\end{theorem}
%

% Theorem \ref{thm-min-consis} also holds for $\mathbb{Q}$-MTAs of a
% fixed alphabet rank, because the minimal consistency problem can be
% reduced to the decision problem for existential first-order sentences
% over~$(\mathbb{Q},+,\cdot,0,1)$ in similar manner to the case of
% words.  Here, fixing the alphabet rank keeps the reduction in
% polynomial time.

%We leave open the complexity of the minimum consistency problem for MTAs over arbitrary alphabets.

\section{Conclusions and Future Work}

We have looked at the computational complexity of computing minimal
multiplicity word and tree automata from several angles.  Specifically,
we have analysed the complexity of computing a minimal automaton
equivalent to a given input automaton $\mathcal{A}$. We have
considered also the corresponding decision problem, which asks whether
there exists an automaton equivalent to $\mathcal{A}$ with a given
number of states.  Finally, we have considered the minimal
  consistency problem, in which the input is a finite set of word-weight
pairs rather than a complete automaton.

Our complexity bounds have drawn connections between automaton
minimisation and longstanding open questions in arithmetic complexity,
including the complexity of polynomial identity testing and the
decidability of Hilbert's tenth problem over the rationals, \ie, the
problem of deciding the truth of existential sentences over the
structure $(\mathbb{Q},+,\cdot,0,1)$.

Our algorithmic results exclusively concern automata over the fields
of rational or real numbers, in which weights are allowed to be
negative.  The minimisation problems considered here all have natural
analogues for the class of \emph{probabilistic automata} over words and
trees, in which the transition weights are probabilities. 
Recently, minimisation of probabilistic word automata was
shown to be \NP-hard \cite{14KW-ICALP}. A natural question is whether
this minimisation problem lies in \NP, and whether the corresponding
problem for tree automata is even harder. Related to this is the
following question: Given a multiplicity (word or tree) automaton with
rational transition weights, need there always be a minimal equivalent
automaton also with rational transition weights?

% Specifically, we have analysed
% complexity of \emph{computing} a minimal automaton 
% equivalent to a given input automaton
% in the unit-cost
% model, of the minimisation \emph{decision problem}, and of the
% \emph{minimal consistency problem}.  One of the key technical
% contributions of our work is Proposition~\ref{prop-compact-space},
% which, based on the product of a given automaton by itself, provides
% small spanning sets for forward space $\forwardspace$ and backward
% space~$\backwardspace$.  This technology also led us to an \NC\
% algorithm for minimising multiplicity \emph{word} automata, thus
% improving the best previous algorithms (polynomial time and
% randomised~\NC).

We have observed that the minimal consistency problem for word automata
over the reals is in PSPACE, since it is directly reducible to the
problem of deciding the truth of existential first-order sentences
over the structure~$(\mathbb{R},+,\cdot,0,1)$.  For tree automata this
reduction is exponential in the alphabet rank, and we leave as an open
question the complexity of the minimal consistency problem for tree
automata over the reals.

 % We have observed that Theorem~\ref{thm-min-consis} holds not only for
 % word automata, but also for tree automata over a fixed alphabet
 % rank. It is an open question whether the complexity of the minimal
 % consistency problem for $\mathbb{F}$-MTAs is higher if the alphabet
 % rank is not fixed.
 
 In all cases, we have considered minimising automata with respect to
 the number of states.  Another natural question is
 minimisation with respect to the number of transitions.  This
 is particularly pertinent to the case of tree automata, where the number
 of transitions is potentially exponential in the number of states.\\

\section*{Acknowledgements.}
The authors would like to thank Michael Benedikt for stimulating discussions, and anonymous referees for their helpful suggestions.
Kiefer is supported by a University Research Fellowship of the Royal Society.
 Maru\v{s}i\'{c} and Worrell gratefully acknowledge the support of the EPSRC.

\bibliographystyle{plain}
\bibliography{references}

%\iftechrep{
%\newpage
%\appendix
%\input{app-linear-algebra}
%\input{appendix}
%\input{app-foundations}
%\input{app-function}
%\input{app-min-consis}
%}{}

\end{document}